\documentclass[journal,10pt]{IEEEtran}
\usepackage{amsmath,amsfonts,amssymb}
\usepackage{algorithmic}
\usepackage{algorithm}
\usepackage{array}
\usepackage{textcomp}
\usepackage{stfloats}
\usepackage{url}
\usepackage{verbatim}
\usepackage{graphicx}
\usepackage{mathtools}
\usepackage{amsthm}
\usepackage{cuted}      
\usepackage{stfloats}   
\usepackage{graphicx}
\newtheorem{proposition}{Proposition}

\newtheorem{remark}{Remark}
\newtheorem{theorem}{Theorem}
\newtheorem{assumption}{Assumption}
\usepackage{booktabs}
\usepackage{graphicx}
\usepackage{makecell}
\usepackage[colorlinks]{hyperref}
\usepackage{subfigure}

\begin{document}

\title{HARQ for Slow Fluid Antenna Multiple Access}

\author{Sixu Han, 
        Kai-Kit Wong,~\IEEEmembership{Fellow,~IEEE}, 
        Hanjiang Hong,~\IEEEmembership{Member,~IEEE}, and 
        Hyundong Shin,~\IEEEmembership{Fellow,~IEEE}
\vspace{-8mm}

\thanks{The work of K. K. Wong is supported by the Engineering and Physical Sciences Research Council (EPSRC) under Grant EP/W026813/1. The work of H. Shin is supported by the National Research Foundation of Korea (NRF) grant funded by the Korean government (MSIT) (RS-2025-00556064 and RS-2025-25442355), and by the Ministry of Science and ICT (MSIT), Korea, under the ITRC (Information Technology Research Center) support program (IITP-2025-RS-2021-II212046), supervised by the IITP (Institute for Information \& Communications Technology Planning \& Evaluation).}

\thanks{S. Han, K. K. Wong, and H. Hong are with the Department of Electronic and Electrical Engineering, University College London, London, United Kingdom. K. K. Wong is also affiliated with the Department of Electronic Engineering, Kyung Hee University, Yongin-si, Gyeonggi-do 17104, Korea (e-mail: \{sixu.han.22, kai-kit.wong, hanjiang.hong\}@ucl.ac.uk).}
\thanks{H. Shin is with the Department of Electronics and Information Convergence Engineering, Kyung Hee University, Yongin-si, Gyeonggi-do 17104, Republic of Korea (e-mail: hshin@khu.ac.kr).}

\thanks{Corresponding author: Kai-Kit Wong.}
}

\markboth{Journal of \LaTeX\ Class Files,~Vol.~14, No.~8, August~2021}%
{Shell \MakeLowercase{\textit{et al.}}: A Sample Article Using IEEEtran.cls for IEEE Journals}


\maketitle

\begin{abstract}
Slow fluid antenna multiple access (\emph{s}FAMA), enabled by the fluid antenna system (FAS), has recently emerged as a practical and low-complexity paradigm for supporting massive wireless connectivity. While existing studies have characterized its physical-layer performance under one-shot transmission, its interaction with retransmission protocols and the resulting networking performance remain largely unexplored. In this paper, we study a downlink hybrid automatic repeat request (HARQ)-assisted \emph{s}FAMA framework, termed HARQ-\emph{s}FAMA, in which each user performs distinguished port selection in every HARQ round and combines the received signals across multiple rounds to improve decoding reliability. We develop a comprehensive analytical framework to characterize the outage probability, average packet waiting time, and energy efficiency of the proposed system. The analysis reveals how HARQ exploits the spatial reconfigurability of FAS to simultaneously enhance reliability and improve queueing performance. Numerical results corroborate the theoretical analysis and demonstrate that the HARQ-\emph{s}FAMA system significantly outperforms conventional one-shot \emph{s}FAMA in terms of reliability, delay, and energy efficiency. These findings suggest that the integration of HARQ and \emph{s}FAMA provides a promising pathway toward a practical and standards-compatible massive access solution for future wireless networks.
\end{abstract}


\vspace{-2mm}
\section{Introduction}
\IEEEPARstart{T}{he rapid} growth of Internet-of-Things (IoT) devices, massive machine-type communications (mMTC), and emerging immersive applications is driving future wireless networks toward an era of unprecedented massive connectivity, while simultaneously imposing increasingly stringent requirements on throughput, latency, and reliability \cite{intro1,intro2,intro4}. To address these challenges, substantial research efforts have been devoted to enhancing multiuser spectral efficiency through massive multiple-input multiple-output (MIMO) \cite{massive1,massive2} and advanced non-orthogonal transmission frameworks such as non-orthogonal multiple access (NOMA) \cite{NOMA, NOMA1} and rate-splitting multiple access (RSMA) \cite{RSMA, RSMA1}. Despite their remarkable gains, however, these schemes generally rely on accurate channel state information (CSI) at the transmitter side and sophisticated processing, which may limit their scalability and practical applicability in massive-access scenarios.

This motivates the search for a low-complexity and scalable multiple access architecture that can support massive spectrum sharing with reduced reliance on CSI. A promising direction in this regard is the fluid antenna system (FAS) concept \cite{FAS1,FAS2,FAS3,spatialblockmodel2,new2025flar,wu2024flu,Lu-2025,hong}, which treats the antenna as a reconfigurable physical-layer resource to broaden system design and optimization, and is highly motivated by advances in reconfigurable and flexible antenna technologies, such as movable elements \cite{zhu2024_fas_history}, liquid-based antennas \cite{shen2024_surfacewave_fas,wang2026_em_reconfig_fas}, pixel-based structures \cite{zhang2025_pixel_reconfig,liu2025_wideband_pixel_fas,wong2026_pixel_meet_fas}, and metamaterial antennas \cite{Zhang-jsac2026,liu2025_meta_fluid_optics}. In \cite{Tong-2025}, various implementation technologies and testbeds for FAS have been discussed and compared.

Since the seminal work in \cite{wong-cl2020fas,FAS}, position reconfigurability in FAS has been recognized as a new spatial degree of freedom (DoF) that can considerably enhance wireless system performance. Building on this concept, fluid antenna multiple access (FAMA) is an umbrella paradigm that advocates the use of FAS for scalable massive multiple access on the same channel use \cite{FAMA3}. FAMA enables massive multiple access without relying on transmitter-side CSI or precoding, thereby reducing signalling overhead and processing complexity. However, it remains compatible with conventional base station (BS)-side precoding techniques when desired. In particular, each user can exploit its FAS to scan the fading conditions across multiple candidate ports and activate those at which interference is most strongly attenuated by multipath propagation. 

Depending on how fast the FAS updates its port, FAMA can be roughly classified into fast FAMA (\emph{f}FAMA) \cite{FAMA1,FAMA2} and slow FAMA (\emph{s}FAMA) \cite{sFAMA,sFAMA2,sFAMA3}. Within these two categories, \emph{f}FAMA performs symbol-level port switching to suppress the instantaneous interference-plus-noise signal, with extraordinary multiplexing capability at the expense of high sensing and switching complexity. In contrast, \emph{s}FAMA updates the selected port only when the CSI changes, offering a more practical implementation with weakened multiplexing gain.

Existing studies on \emph{s}FAMA mainly focused on its block-level interference mitigation capability and single-transmission physical-layer performance, particularly in terms of one-shot outage behavior. In practice, hybrid automatic repeat request (HARQ) is a mechanism that combines forward error correction with retransmission-based error recovery. It has become a fundamental reliability-enhancement mechanism in wireless networks and has been extensively studied \cite{HARQ1}. In particular, under HARQ with chase combining (HARQ-CC), repeated transmissions of the same packet can be accumulated and combined at the receiver across multiple rounds \cite{HARQCC}. The HARQ-based reliability enhancement has also been demonstrated in advanced multiple access frameworks such as RSMA \cite{1,2,3} and NOMA \cite{4,5}. Nevertheless, HARQ-assisted \emph{s}FAMA has not been studied before. This naturally motivates a compatible HARQ design to further enhance the reliability of \emph{s}FAMA networks. Moreover, a fundamental performance analysis of HARQ-assisted \emph{s}FAMA is cucial under an information-theoretic decoding framework, particularly with respect to the reliability-latency-throughput trade-off. 

Motivated by this, this paper proposes a downlink HARQ-assisted \emph{s}FAMA framework, namely HARQ-\emph{s}FAMA, in which each user is equipped with a FAS to perform \emph{s}FAMA in each HARQ round. We characterize the reliability, delay, throughput, and energy efficiency performance of the HARQ-\emph{s}FAMA network under an information-theoretic decoding criterion, and validate the analysis through numerical results. 
Specifically, our main contributions are summarized as follows:
\begin{itemize}
\item We propose the HARQ-\emph{s}FAMA framework, wherein each user equipment (UE) performs \emph{s}FAMA in each HARQ round. Additionally, per-UE HARQ-CC is implemented for the active packet, with the signal being re-transmitted over an independent block-fading channel. Each UE is equipped with a FAS to select the port  that maximizes the per-round signal-to-interference ratio (SIR), and it combines the received signals from multiple HARQ rounds using the proposed optimal combining method. 
\item We present a theoretical analysis of the performance of the HARQ-\emph{s}FAMA framework. Specifically, we employ the spatial block-correlation model in \cite{spatialblockmodel,spatialblockmodel1}, and derive analytical and semi-analytical expressions for key performance metrics, including the outage probability, average waiting time, payload throughput, and energy efficiency, thereby enabling a systematic characterization of the trade-off among reliability, delay, throughput, and energy efficiency in the HARQ-\emph{s}FAMA framework.
\item Simulation results are provided to assess HARQ-\emph{s}FAMA and validate our analytical results. The results reveal a reliability-delay-energy-efficiency trade-off: first, increasing the SIR threshold or the HARQ round limit improves reliability but increases the retransmission burden and degrades delay and energy efficiency under heavy traffic. They also show that the proposed analytical framework is accurate when FAS port numbers are not particularly large. However, the outage and system-level gains saturate at high port densities due to spatial correlation over the fixed-size aperture, and FAS consistently outperforms the fixed-position antenna (FPA) baseline.
\end{itemize}


\vspace{-2mm}
\section{\emph{s}FAMA Network Model} \label{NetMod} 
In this paper, we consider a downlink multiuser network where a BS with $N_t$ FPAs serves $U$ UEs over the same channel use. The BS antennas are sufficiently far apart to be spatially independent. Each antenna is responsible for transmitting an information-bearing signal to a designated UE, i.e., $N_t = U$. Each UE is equipped with a one-dimensional FAS (1D-FAS) of physical size $ W\lambda $, where $\lambda$ is the carrier wavelength. The 1D-FAS has $K$ ports uniformly distributed over the space. 

For UE~$u$, the information data is divided into $Q$ packets, with the $q$-th packet denoted by $\mathbf d_u^{(q)} \in \{0,1\}^{M_{\mathrm s}}$, for $q=1,\ldots,Q$, and $M_{\mathrm s}$ being the number of information bits per packet. Each packet $\mathbf d_u^{(q)}$ is mapped into a length-$L$ block of symbols $\mathbf s_u^{(q)} = \bigl[s_{u,1}^{(q)},\,s_{u,2}^{(q)},\,\ldots,\,s_{u,L}^{(q)}\bigr]^{\mathsf T} \in \mathbb C^L$, which is transmitted over one quasi-static fading block. Accordingly, we define the initial spectral efficiency per packet as \(R_0 \triangleq M_{\mathrm s}/L\) bits/channel use. The average symbol energy is defined as $E_s \triangleq \frac{1}{L}\,\mathbb E\!\left[\bigl\|\mathbf s_u^{(q)}\bigr\|^2\right]$. Similarly, $E_{\tilde u} \triangleq \frac{1}{L}\,\mathbb E\!\left[\|\mathbf s_{\tilde u}\|^2\right]$ for the $\tilde{u}$-th UE. Under equal-power transmission, we have $E_{\tilde u}=E_s$. Packet-header overhead is usually very small compared to the packet size and thus neglected throughout the analysis. 

Suppose that the packet $\mathbf s_u^{(q)}$ is transmitted in the $n$-th block. The received signal at the $k$-th port of UE $u$ is modeled as
\begin{equation}\label{eq:ruk}
\mathbf{r}_{u,k}^{(q)}=  g_{(u,u),k}^{(n)} \mathbf s_u^{(q)}+ \sum_{\substack{\tilde u = 1 \\ \tilde u \neq u}}^{U}+g_{(\tilde u,u),k}^{(n)} \mathbf s_{\tilde u}+ \boldsymbol\eta_{u,k}^{(n)}.
\end{equation}
where $g_{(\tilde u,u),k}^{(n)} \in \mathbb C$ denotes the effective block-fading channel coefficient associated with the BS signal intended for UE~$\tilde u$ as observed at the $k$-th port of UE~$u$ in fading block $n$, and $\mathbf g_{(\tilde u,u)}^{(n)}\triangleq[g_{(\tilde u,u),1}^{(n)},\ldots,g_{(\tilde u,u),K}^{(n)}]^{\mathsf T}\in \mathbb C^{K}$ is the channel vector across the $K$ ports. The $q$-th packet intended for UE~$u$ is denoted by $\mathbf s_u^{(q)}$ and corresponds to the $q$-th transmitted symbol sequence among the $Q$ subsequences, whereas $\mathbf s_{\tilde u}$ is the length-$L$ symbol sequence transmitted for user $\tilde u$ in the considered fading block, with its packet index not explicitly tracked. The transmitted symbol blocks are assumed independent across fading blocks and independent of the small-scale fading realizations. Moreover, $\boldsymbol\eta_{u,k}^{(n)}\in\mathbb C^L$ denotes the additive white Gaussian noise (AWGN) vector at the $k$-th port of UE~$u$ in fading block $n$, whose entries are independent and identically distributed  (i.i.d.) with zero mean and variance $\sigma_{\eta}^{2}$.

Regarding $\mathbf g_{(\tilde u,u)}^{(n)}$, the channel coefficients can be spatially correlated, since the $K$ ports may be closely spaced. Consider a rich-scattering environment where $\mathbf g_{(\tilde u,u)}^{(n)}$ follows a jointly Gaussian distribution \cite{spatialblockmodel}, i.e.,
\begin{equation}\label{eq:blk_joint_gauss}
\mathbf g_{(\tilde u,u)}^{(n)} \sim \mathcal{CN}\!\big(\mathbf 0,\ \sigma^2 \boldsymbol{\Sigma}\big),
~ \forall n=1,\ldots,N, \forall \tilde u = 1, \dots, U.
\end{equation}
where $\boldsymbol{\Sigma}$ represents the spatial correlation matrix with each element defined as $[\boldsymbol{\Sigma}]_{k,l} = J_0(2\pi (k-l)W/(K-1))$. Here, $J_0(\cdot)$ is the Bessel function of the first kind.

Assuming interference-limited scenario, the $u$-th UE using \emph{s}FAMA selects the port $k_u^{(n)}$ maximizing the SIR, i.e.
\begin{equation}\label{k_u^{(n)}}
k_u^{(n)} =  \arg\max_{k=1,\ldots,K} \mathrm{SIR}_{u,k}^{(n)}
\end{equation}
where
\begin{equation}\label{eq:SIR_uk}
    \mathrm{SIR}_{u,k}^{(n)} \triangleq \frac{|g_{(u,u),k}^{(n)}|^2 E_s}
         {\sum_{\substack{\tilde u = 1 \\ \tilde u \neq u}}^{U}|g_{(\tilde u,u),k}^{(n)}|^2 E_{\tilde u}} 
         \stackrel{(a)}{=} \frac{|g_{(u,u),k}^{(n)}|^2}
         {\sum_{\substack{\tilde u = 1 \\ \tilde u \neq u}}^{U}|g_{(\tilde u,u),k}^{(n)}|^2}, 
\end{equation} 
where $(a)$ follows the equal-power assumption across users.

Performance analysis under the correlation models in \eqref{eq:blk_joint_gauss} is challenging. To facilitate mathematically tractable analysis, we employ the spatial block-correlation model \cite{spatialblockmodel} to approximate $\boldsymbol{\Sigma}$. In this model, $\boldsymbol{\Sigma}$ can be approximated by a block-diagonal matrix $\widehat{\boldsymbol{\Sigma}}$ defined as
\begin{equation}\label{eq:Sigma_hat_blockdiag}
\widehat{\boldsymbol{\Sigma}} \in \mathbb R^{K\times K}
=
\begin{pmatrix}
\mathbf A_1 & \mathbf 0  & \cdots & \mathbf 0 \\
\mathbf 0   & \mathbf A_2& \cdots & \mathbf 0 \\
\vdots      & \ddots     & \ddots & \vdots \\
\mathbf 0   & \mathbf 0  & \cdots & \mathbf A_B
\end{pmatrix}.
\end{equation}
where each diagonal block $\mathbf A_b$, $b\in\{1,2,\ldots,B\}$, is a constant correlation matrix of size $L_b\times L_b$, satisfying $\sum_{b=1}^{B}L_b=K$. 
Each block, $\mathbf A_b$, admits the following common form:
\begin{equation}\label{eq:Ab_common_form}
\mathbf A_b
=\begin{pmatrix}
1        & \mu_b^{2} & \cdots   & \mu_b^{2} \\
\mu_b^{2}& 1         & \cdots   & \mu_b^{2} \\
\vdots   & \ddots    & \ddots   & \vdots \\
\mu_b^{2}& \cdots    & \mu_b^{2}& 1
\end{pmatrix},
\end{equation}
where $\mu_b^2 \in (0.95,0.99)$ denotes the within-block spatial correlation constant associated with the $b$-th diagonal block $\mathbf A_b$. For analytical tractability, we further assume a common within-block correlation constant across all spatial blocks, i.e., $\mu_b=\mu, \forall b=1,\ldots,B$. In the subsequent analysis, the exact spatial correlation matrix $\boldsymbol{\Sigma}$ is approximated by the block-diagonal matrix $\widehat{\boldsymbol{\Sigma}}$, i.e., we assume $\mathbf{g}^{(n)}_{(\tilde u,u)}\sim \mathcal{CN}(\boldsymbol{0},\sigma^2\widehat{\boldsymbol{\Sigma}})$.

Therefore, $\{g_{(\tilde u,u),k}^{(n)}\}_{\forall k}$ can be represented by $B$ pairs of common real-valued Gaussian variables $\{(\breve{x}_{\tilde u,b}^{(n)},\breve{y}_{\tilde u,b}^{(n)})\}_{b=1}^{B}$ and independent port-specific pairs $\{(x_{\tilde u,k}^{(n)},y_{\tilde u,k}^{(n)})\}_{k=1}^{K}$. Let $b(k)\in\{1,\ldots,B\}$ denote the group index of port $k$, i.e., $b(k)=1$ for $k=1,\ldots,L_1$, $b(k)=2$ for $k=L_1+1,\ldots,L_1+L_2$, and so on. The channel $g_{(\tilde u,u),k}^{(n)}$ is then approximated by
\begin{multline}\label{eq:common_gauss_multiblock}
g_{(\tilde u,u),k}^{(n)} \approx \sigma\Big(\sqrt{1-\mu^{2}}\,x_{\tilde u,k}^{(n)}+\mu\,\breve{x}_{\tilde u,b(k)}^{(n)}\Big)\\
+ \mathrm{j}\,\sigma\Big(\sqrt{1-\mu^{2}}\,y_{\tilde u,k}^{(n)}+\mu\,\breve{y}_{\tilde u,b(k)}^{(n)}\Big). 
\end{multline}
where all real Gaussian random variables are independent with zero mean and variance $1/2$. The SIR in \eqref{eq:SIR_uk} is written as
\begin{equation}\label{eq:SIR_ratio_XY}
\mathrm{SIR}_{u,k}^{(n)} \approx \frac{S_{u,k}^{(n)}}{I_{u,k}^{(n)}}.
\end{equation}
where $S_{u,k}^{(n)}$ and $I_{u,k}^{(n)}$ are defined as
\begin{align}
S_{u,k}^{(n)} &\triangleq \left(x_{u,k}^{(n)}+\frac{\mu}{\sqrt{1-\mu^2}}x_{u,b(k)}^{(n)}\right)^2 \nonumber\\
&\quad\quad\quad\quad + \left(y_{u,k}^{(n)}+\frac{\mu}{\sqrt{1-\mu^2}}y_{u,b(k)}^{(n)}\right)^2, \label{eq:X_ukn_def}\\
I_{u,k}^{(n)} &\triangleq \sum_{\substack{\tilde u=1\\\tilde u\neq u}}^{U} \left(x_{\tilde u,k}^{(n)}+\frac{\mu}{\sqrt{1-\mu^2}}x_{\tilde u,b(k)}^{(n)}\right)^2 \nonumber\\
&\quad\quad\quad\quad+ \left(y_{\tilde u,k}^{(n)}+\frac{\mu}{\sqrt{1-\mu^2}}y_{\tilde u,b(k)}^{(n)}\right)^2.\label{eq:I_u_k_n_sum}
\end{align}

Only the selected port is connected to the single RF chain. The received SIR of UE~$u$ in block~$n$ for the \emph{s}FAMA network is defined as $\mathrm{SIR}_{u}^{(n)}\triangleq \mathrm{SIR}_{u,k_u^{(n)}}^{(n)}$, where $k_u^{(n)}$ is given by \eqref{k_u^{(n)}}. Consequently, we have
\begin{equation}\label{eq:SIR_max_relation}
\mathrm{SIR}_{u}^{(n)}=\max_{k=1,\ldots,K}\mathrm{SIR}_{u,k}^{(n)}=\max_{k=1,\ldots,K}\frac{S_{u,k}^{(n)}}{I_{u,k}^{(n)}}.
\end{equation}

\vspace{-2mm}
\section{Proposed Retransmission Protocol}
\label{Proposed Retransmission Protocol}
In this section, we propose the packet-level HARQ protocol built on the \emph{s}FAMA network, including the round/block indexing, the received-signal model for each HARQ round, and the receiver combining and termination rules.

\vspace{-2mm}
\subsection{Packet-Level HARQ-\emph{s}FAMA}
\label{subsec:packet_level_harq}
In the HARQ-\emph{s}FAMA system, HARQ is performed at the packet level for each UE. Consistent with the transmission model in Section~\ref{NetMod}, each HARQ round occupies one quasi-static fading block and carries one coded-modulated replica of the active packet. We adopt a per-UE stop-and-wait HARQ discipline on the considered time-frequency resource, where each UE maintains a single active head-of-line (HOL) packet and no parallel HARQ processes are allowed for the same UE. Specifically, once packet $\mathbf{s}_u^{(q)}$ becomes active, it remains the unique packet under service for UE~$u$ until it is either successfully decoded or declared in outage after at most $C$ HARQ rounds, while all other packets are buffered in a First-In, First-Out (FIFO) queue. We index HARQ rounds by $i\in\{1,\ldots,C\}$, where $i=1$ denotes the initial transmission, and $C$ represents the maximum number of rounds.

To index HARQ transmissions on the global fading-block timeline, let $n_{u,q,i}\in\mathbb N_{\ge 1}$ denote the fading-block index in which round $i$ of packet $q$ for UE~$u$ is transmitted.  Under the one-transmission-per-block constraint, $n_{u,q,i}$ specifies the absolute location of the $i$-th HARQ round of packet $q$ on the global block axis, whereas $i$ itself records only the within-packet HARQ round index. As the subsequent analysis focuses on a specific tagged packet, the global block indices can be omitted for simplicity. Accordingly, for any block-dependent quantity $x^{(n)}$, we define $x^{(i)} \triangleq x^{(n_{u,q,i})}$, to denote its value in the fading block occupied by the round $i$ of packet $q$ for UE~$u$. Similarly, we can also define
\begin{equation}
\left\{\begin{aligned}
g_{\tilde u,u,k}^{(i)} &\triangleq g_{\tilde u,u,k}^{(n_{u,q,i})},\\
k_u^{(i)} &\triangleq k_u^{(n_{u,q,i})},\\
I_{u,k}^{(i)} &\triangleq I_{u,k}^{(n_{u,q,i})},
\end{aligned}\right.
\text{and }\left\{\begin{aligned}
\boldsymbol\eta_{u,k}^{(i)} &\triangleq \boldsymbol\eta_{u,k}^{(n_{u,q,i})},\\
S_{u,k}^{(i)} &\triangleq S_{u,k}^{(n_{u,q,i})},\\
\mathrm{SIR}_{u,k}^{(i)} &\triangleq {\rm SIR}_{u,k}^{(n_{u,q,i})}.
\end{aligned}\right.
\end{equation}

For the $i$-th HARQ round, the $u$-th UE received signal at the $k_u^{(i)}$-th port according to the \emph{s}FAMA reception in \eqref{k_u^{(n)}}. The selected-port received signal can be written as
\begin{equation}\label{eq:ruqi_rx}
\mathbf r_{u,i}^{(q)}=g_{(u,u),k_u^{(i)}}^{(i)}\,\mathbf s_u^{(q)}+\sum_{\substack{\tilde u=1\\ \tilde u\neq u}}^{U}a_{\tilde u}^{(i)}\, g_{(\tilde u,u),k_u^{(i)}}^{(i)}\,\mathbf s_{\tilde u}^{(i)}+\boldsymbol\eta_{u,k_u^{(i)}}^{(i)},
\end{equation}
in which $\mathbf s_{\tilde u}^{(i)}\in\mathbb C^L$ denotes the symbol sequence transmitted by the BS for UE~$\tilde u$ in the fading block used by round $i$, and the binary indicator $a_{\tilde u}^{(i)}\in\{0,1\}$ specifies whether the stream intended for UE~$\tilde u$ is active on the considered time-frequency resource in round $i$: $a_{\tilde u}^{(i)}=1$ means that UE~$\tilde u$ is active and hence contributes interference to UE~$u$, whereas $a_{\tilde u}^{(i)}=0$ means that no interfering stream for UE~$\tilde u$ is present in that round. To proceed, we define the number of active interferers in the $i$-th HARQ round as
\begin{equation}\label{eq:A_i_def_sys}
A^{(i)} \triangleq \sum_{\substack{\tilde u=1\\ \tilde u\neq u}}^{U} a_{\tilde u}^{(i)}.
\end{equation}

As this work focuses on the interference-limited regime, we condition the per-round SIR analysis on the event $A^{(i)}\ge 1$, where at least one co-scheduled interferer is active in the considered HARQ round. The event $A^{(i)}=0$ corresponds to an interference-free round, for which the FAMA port selection remains well defined but is no longer governed by inter-user interference. This case is therefore outside the interference-limited SIR distribution analyzed below and can be incorporated separately if required. Accordingly, the activity count used in the following analysis is modeled by the zero-truncated distribution conditioned on $A^{(i)}\ge 1$. The probabilistic model for the activity indicators $\{a_{\tilde u}^{(i)}\}$, and hence for $A^{(i)}$, will be specified in the performance analysis in Section \ref{subsec:outage_analysis}. 

The selected port indices $k_u^{(i)}$ are independent across HARQ rounds, resulting in a sequence of independent SIRs $\{\mathrm{SIR}_{u}^{(i)}\}$. At each HARQ round, the port selection decision also depends on the round-specific activity indicators $\{a_{\tilde u}^{(i)}\}$. As a result, the selected port in the $i$-th HARQ round is given by
\begin{equation}\label{eq:ku_i_HARQ_compact}
k_u^{(i)}=\arg\max_{k}\frac{|g_{(u,u),k}^{(i)}|^2}{\sum_{\substack{\tilde u=1\\ \tilde u\neq u}}^{U}a_{\tilde u}^{(i)} |g_{(\tilde u,u),k}^{(i)}|^2 }=\arg\max_{k}\frac{S_{u,k}^{(i)}}{I_{u,k}^{(i)}},
\end{equation}
where $S_{u,k}^{(i)}$ follows from \eqref{eq:X_ukn_def}, and $I_{u,k}^{(i)}$ is now defined as
\begin{equation}\label{eq:I_u_k_i_sum_short}
I_{u,k}^{(i)}
\triangleq
\sum_{\substack{\tilde u=1\\ \tilde u\neq u}}^{U}
a_{\tilde u}^{(i)}\, S_{\tilde u,k}^{(i)}.
\end{equation}

\vspace{-2mm}
\subsection{HARQ Operation and Receiver Combining}
We adopt per-UE HARQ-CC for the active packet. For a given decoding attempt after the $i$-th HARQ round, where $i\in\{1,\ldots,C\}$, we use $j\in\{1,\ldots,i\}$ to index the HARQ rounds whose observations are combined. The receiver then forms the effective combined observation as
\begin{equation}\label{eq:combiner_buffer_def}
\mathbf z_{u,i}^{(q)}
\triangleq
\sum_{j=1}^{i}\big(\alpha_{u,j}^{(q)}(i)\big)^*\,\mathbf r_{u,j}^{(q)},
\end{equation}
where $\alpha_{u,j}^{(q)}(i)\in\mathbb C$ denotes the combining coefficient applied to the observation from HARQ round $j$ when decoding is attempted after the $i$-th round. For convenience, we define the combining vector as $\boldsymbol{\alpha}_{u}^{(q)}(i)\triangleq[\alpha_{u,1}^{(q)}(i),\ldots,\alpha_{u,i}^{(q)}(i)]^{\mathsf T}\in\mathbb C^{i}$, which will be specified in Section~\ref{subsec:sir_opt_combining}. 

UE $u$ then attempts to decode the packet based on $\mathbf z_{u,i}^{(q)}$. If decoding is successful, the UE sends an ACK to the BS, the packet departs service, and the observations associated with that packet are removed from the combining buffer. Otherwise, a NACK is sent back to the BS. In this case, if $i<C$, the packet remains active and all observations accumulated up to round $i$, namely $\{\mathbf r_{u,j}^{(q)}\}_{j=1}^{i}$, are retained for combining with future retransmissions; if $i=C$, the packet is declared in outage and all stored observations are discarded. For analytical tractability, the ACK/NACK feedback is assumed error-free and instantaneous, and its transmission and processing delay is neglected. Consequently, the feedback does not consume an additional data block in the adopted model.

\vspace{-2mm}
\subsection{SIR-Optimal Soft Combining}
\label{subsec:sir_opt_combining}
Regarding the decoding attempt after the $i$-th round, the port selections $\{k_u^{(j)}\}_{j=1}^{i}$ are carried out in accordance with \eqref{eq:ku_i_HARQ_compact}, sequentially and in a causal sequence. Based on the selected ports and the observations accumulated up to round $i$, the UE receiver then employs a linear combining vector $\boldsymbol{\alpha}_{u}^{(q)}(i)$ to maximize the combined SIR. As a result, the combination design described in this section should be regarded as a post-processing problem under causal conditions, rather than a non-causal joint optimization problem across all HARQ rounds.

Substituting \eqref{eq:ruqi_rx} into \eqref{eq:combiner_buffer_def} yields
\begin{align}
\mathbf z_{u,i}^{(q)}&=\left(\sum_{j=1}^{i}\big(\alpha_{u,j}^{(q)}(i)\big)^*\,g_{(u,u),k_u^{(j)}}^{(j)}\right)\mathbf s_u^{(q)}
\nonumber\\
&\quad+
\sum_{j=1}^{i}\big(\alpha_{u,j}^{(q)}(i)\big)^*
\sum_{\tilde u\neq u} a_{\tilde u}^{(j)}\,
g_{(\tilde u,u),k_u^{(j)}}^{(j)}\mathbf s_{\tilde u}^{(j)}
\nonumber\\
&\quad+
\sum_{j=1}^{i}\big(\alpha_{u,j}^{(q)}(i)\big)^*\,
\boldsymbol{\eta}_{u,k_u^{(j)}}^{(j)} .
\label{eq:z_expand_full_main}
\end{align}
Under the HARQ-CC protocol, the desired component combines coherently across rounds. We assume an the interference-limited regime where noise is neglected. To obtain a tractable post-combining interference-power expression, we introduce the following local assumption.

\begin{assumption}[Interference uncorrelatedness across HARQ rounds]\label{as:temp_uncorr_I}
Let
\begin{equation}\label{eq:iota_def_main}
\boldsymbol{\iota}_{u,j}^{(q)}
\triangleq
\sum_{\substack{\tilde u=1\\ \tilde u\neq u}}^{U}
a_{\tilde u}^{(j)}\, g_{(\tilde u,u),k_u^{(j)}}^{(j)}\,\mathbf s_{\tilde u}^{(j)}
\in\mathbb C^L
\end{equation}
denote the aggregate interference at the selected port in the $j$-th HARQ round for the tagged packet $q$ of UE~$u$. We assume
\begin{equation}\label{eq:cross_round_I_uncorr}
\frac{1}{L}\,
\mathbb E\!\left[
\big(\boldsymbol{\iota}_{u,j}^{(q)}\big)^{\mathsf H}
\boldsymbol{\iota}_{u,\ell}^{(q)}
\right]
=0,~\forall j\neq \ell,
\end{equation}
which is consistent with the adopted i.i.d.\ block-fading model and roundwise randomization of the interfering transmissions.
\end{assumption}

For a given decoding attempt after round $i$, and using the selected port $k_u^{(j)}$ in each round $j\in\{1,\ldots,i\}$, the activity-aware interference power associated with the selected port in round $j$ is given by
\begin{equation}
\sum_{\substack{\tilde u=1\\ \tilde u\neq u}}^{U}
a_{\tilde u}^{(j)}\,\big|g_{(\tilde u,u),k_u^{(j)}}^{(j)}\big|^2
=
\sigma^2(1-\mu^2)\,I_{u,k_u^{(j)}}^{(j)},
\end{equation}
where $I_{u,k_u^{(j)}}^{(j)}$ is given by \eqref{eq:I_u_k_i_sum_short} with $k=k_u^{(j)}$. 

Under Assumption~\ref{as:temp_uncorr_I}, the cross terms in the average combined interference power vanish. Thus, if the packet remains active up to round $i$, the combined SIR at the $i$-th decoding attempt can be approximated as
\begin{equation}\label{eq:post_SIR_main}
\mathrm{SIR}_{u}^{(q)}\!\left(i;\boldsymbol{\alpha}_{u}^{(q)}(i)\right)
\approx
\frac{
\Big|\sum_{j=1}^{i}\big(\alpha_{u,j}^{(q)}(i)\big)^*\, g_{(u,u),k_u^{(j)}}^{(j)}\Big|^{2}
}{
\sigma^2(1-\mu^2)\sum_{j=1}^{i}\big|\alpha_{u,j}^{(q)}(i)\big|^{2}\, I_{u,k_u^{(j)}}^{(j)}
}.
\end{equation}
For compactness, we define
\begin{equation}\label{eq:h_D_def_main}
\begin{aligned}
\mathbf g_{u}^{(q)}(i)&\triangleq\left[g_{(u,u),k_u^{(1)}}^{(1)},\ldots,g_{(u,u),k_u^{(i)}}^{(i)}\right]^{\mathsf T},\\
\mathbf D_{u}^{(q)}(i)&\triangleq\operatorname{diag}\!\left(I_{u,k_u^{(1)}}^{(1)},\ldots,I_{u,k_u^{(i)}}^{(i)}\right).
\end{aligned}
\end{equation}
Then \eqref{eq:post_SIR_main} can be rewritten as
\begin{equation}\label{eq:post_SIR_vec_main}
\mathrm{SIR}_{u}^{(q)}\!\left(i;\boldsymbol{\alpha}_{u}^{(q)}(i)\right)\approx\frac{\big|\big(\boldsymbol{\alpha}_{u}^{(q)}(i)\big)^{\mathsf H}\mathbf g_{u}^{(q)}(i)\big|^{2}}{\sigma^2(1-\mu^2)\,\boldsymbol{\alpha}_{u}^{(q)}(i)^{\mathsf H}\mathbf D_{u}^{(q)}(i)\,\boldsymbol{\alpha}_{u}^{(q)}(i)}.
\end{equation}
In the considered HARQ rounds, given the constraint $A^{(j)}\ge 1$, it necessarily follows that $I_{u,k_u^{(j)}}^{(j)}>0, \forall j=1,\ldots,i$. Hence, $\mathbf D_u^{(q)}(i)$ is invertible.

\begin{proposition}[SIR-optimal soft combining]\label{prop:sir_opt_combining}
In interference-limited scenarios, under Assumption~\ref{as:temp_uncorr_I}, an SIR-maximizing combining vector with packet and HARQ round indices omitted is given by
\begin{equation}\label{eq:alpha_opt_vec_main}
\boldsymbol{\alpha} ^\star_u=
c \big(\mathbf D_{u}^{(q)}(i)\big)^{-1}\mathbf g_{u}^{(q)}(i),
\end{equation}
where $c$ is an arbitrary non-zero complex scalar. With the optimal combining vector, the SIR of the HARQ-\emph{s}FAMA network after $i$ HARQ rounds for packet~$q$ can be derived as
\begin{equation}\label{eq:Gamma_def_main}
\begin{aligned}
\Gamma_u^{(q)}(i)
\!\triangleq
\mathrm{SIR}_{u}^{(q)}\!\left(i;\boldsymbol{\alpha} ^\star_u\right) 
= \sum_{j=1}^{i}
\frac{S_{u,k_u^{(j)}}^{(j)}}{I_{u,k_u^{(j)}}^{(j)}}
\triangleq
\sum_{j=1}^{i}\mathrm{SIR}_u^{(j)},
\end{aligned}
\end{equation}
\end{proposition}

\begin{proof}
See Appendix~\ref{app:proof_prop_sir_opt}.
\end{proof}

Proposition~\ref{prop:sir_opt_combining} shows that the combined SIR after $i$ HARQ rounds is given by the sum of the round-wise selected-SIR terms in \eqref{eq:Gamma_def_main}. Therefore, conditioning on the realizations of the first $i-1$ HARQ rounds, the combined SIR at the $i$-th decoding attempt depends on the current-round port only through the ratio ${S_{u,k}^{(i)}}/{I_{u,k}^{(i)}}$. Hence, maximizing the combined SIR over the current-round port reduces to the standard per-round \emph{s}FAMA rule. Since the selected-port SIRs are also independent across HARQ rounds, the outage analysis reduces to characterizing the distribution of the per-round selected SIR and then obtaining the distribution of $\Gamma_u^{(q)}(i)$ via repeated convolution. We will then formalize the decoding criterion and derive the resulting outage probability in Section \ref{Analysis of The Proposed Approach}.

\begin{remark}\label{rem:model_scope_main}
The combiner structure in Proposition~\ref{prop:sir_opt_combining} and the above consistency interpretation both rely on the interference-limited approximation and Assumption~\ref{as:temp_uncorr_I}. If noise is non-negligible or the interference across HARQ rounds is temporally correlated, the additive SIR structure in \eqref{eq:Gamma_def_main} may break down, and the stated conclusions need not hold in general.
\end{remark}

\vspace{-2mm}
\section{Performance Analysis}\label{Analysis of The Proposed Approach}
This section develops the analytical framework to evaluate the efficiency of the proposed HARQ-\emph{s}FAMA scheme. 
We first derive the outage probability under the HARQ-CC protocol, and then study the average number of transmission rounds, the average waiting time, the resource-normalized payload throughput, and the corresponding energy efficiency. 

\vspace{-2mm}
\subsection{Outage Probability}\label{subsec:outage_analysis}
Based on the effective combined SIR in \eqref{eq:Gamma_def_main}, the instantaneous achievable rate after $i$ HARQ rounds is obtained by
\begin{equation}
\label{eq:MI_def_main}
I_u^{(q)}(i)
\triangleq
\log_2\!\bigl(1+\Gamma_u^{(q)}(i)\bigr).
\end{equation}
Since HARQ-CC operates by soft combining the received replicas of the same coded packet, $I_u^{(q)}(i)$ characterizes the mutual information supported by the resulting equivalent channel after $i$ transmission rounds. To facilitate our analysis, we adopt the standard threshold decoding abstraction in which decoding after round~$i$ is declared successful whenever
\begin{equation}\label{eq:decode_rule_main}
I_u^{(q)}(i)\ge R_0,
\end{equation}
or equivalently,
\begin{equation}\label{eq:gamma_th_main}
\Gamma_u^{(q)}(i)\ge \gamma_{\mathrm{th}},~\mbox{where }\gamma_{\mathrm{th}}\triangleq 2^{R_0}-1.
\end{equation}
Hence, the outage probability after at most $C$ HARQ rounds can be defined as
\begin{equation}\label{eq:Pout_def_main}
P_{\mathrm{out}}(C)\triangleq\mathbb P\bigl(\Gamma_u^{(q)}(C)<\gamma_{\mathrm{th}}\bigr).
\end{equation}
To simplify the notations, we denote $P_{\mathrm{out}}\equiv P_{\mathrm{out}}(C)$ when the maximum number of HARQ rounds $C$ is fixed.

To evaluate \eqref{eq:Pout_def_main}, we first characterize the distribution of ${\rm SIR}_u^{(i)}$. With the round-wise \emph{s}FAMA port selection method in \eqref{eq:ku_i_HARQ_compact}, the SIR for the $i$-th HARQ round can be obtained as 
\begin{equation}\label{eq:SIR_selected_main}
\mathrm{SIR}_u^{(i)}=\max_{k=1,\ldots,K}\frac{S_{u,k}^{(i)}}{I_{u,k}^{(i)}},
\end{equation}
where $S_{u,k}^{(i)}$ follows \eqref{eq:X_ukn_def} and $I_{u,k}^{(i)}$ is given by \eqref{eq:I_u_k_i_sum_short}.

We first consider the distribution of the number of active interferers observed during the specified HARQ round, $A^{(i)}$, as defined in \eqref{eq:A_i_def_sys}. Under the mean-field Bernoulli activity model, the indicators $\{a_{\tilde u}^{(i)}\}$ are assumed i.i.d.\ with activity probability $p_{\mathrm a}$, independently across $\tilde u$ and $i$. Hence, $A^{(i)}$ is binomially distributed and has the same law for every $i$. Considering an interference-limited environment, our analysis is restricted to the non-trivial scenario where $A^{(i)}\ge 1$, and the corresponding zero-truncated activity count is defined as 
\begin{equation}
\label{eq:Abar_def_main}
\bar A^{(i)}
\triangleq
A^{(i)}\mid\left\{A^{(i)}\ge 1\right\}.
\end{equation}
Accordingly, for \(m=1,\ldots,U-1\), \(\bar A^{(i)}\) follows the zero-truncated binomial distribution, with its probability mass function (pmf) given by
\begin{equation}\label{eq:Abar_pmf_main}
\mathbb P\bigl(\bar A^{(i)}=m\bigr)=\binom{U-1}{m} \times \frac{p_{\mathrm a}^{m}(1-p_{\mathrm a})^{U-1-m}}{1-(1-p_{\mathrm a})^{U-1}}.
\end{equation}

Since the HARQ rounds are statistically identical, the conditional distribution of the selected per-round SIR is the same for every $i$. Omitting the round index, we denote the generic conditional cumulative distribution function (CDF) of the per-round SIR given $\bar A=m$ as $F_{\mathrm{SIR}_u}(\gamma\mid \bar A=m)$.

\begin{proposition}[Conditional CDF of the per-round SIR]\label{prop:cond_cdf_sir_main}
For any $\gamma\ge 0$ and $m\in\{1,\ldots,U-1\}$, the conditional CDF of the per-round \emph{s}FAMA SIR of UE~$u$ is given by
\begin{multline}\label{eq:FSIR_conditional_main}
F_{\mathrm{SIR}_u}\!(\gamma \!\mid\! \bar A\!=\!m)\!=\prod_{b=1}^{B}\int_{0}^{\infty}\!\!\int_{0}^{\infty}
\frac{\tilde r_{u,b}^{\,m-1}e^{-(r_{u,b}+\tilde r_{u,b})}}{\Gamma(m)}\\
\times\!\bigl[\Xi_m\!(\gamma;r_{u,b},\tilde r_{u,b})\bigr]^{L_b}\,dr_{u,b}\,d\tilde r_{u,b},
\end{multline}
where $\Xi_m(\gamma;r_{u,b},\tilde r_{u,b})$ is given in \eqref{eq:Xi_main}, see top of this page.
\end{proposition}

\begin{proof}
See Appendix~\ref{app:per_round_sir_proof}.
\end{proof}

\begin{figure*}[t]
\begin{equation}
\label{eq:Xi_main}
\begin{aligned}
\Xi_m(\gamma;r_{u,b},\tilde r_{u,b})
&=
Q_m\!\left(
\frac{\mu}{\sqrt{1-\mu^2}}
\sqrt{\frac{2\gamma \tilde r_{u,b}}{\gamma+1}},
\frac{\mu}{\sqrt{1-\mu^2}}
\sqrt{\frac{2 r_{u,b}}{\gamma+1}}
\right)
-
\left(\frac{1}{\gamma+1}\right)^m
\exp\!\left(
-\frac{\mu^2}{1-\mu^2}\,
\frac{\gamma \tilde r_{u,b}+r_{u,b}}{\gamma+1}
\right) \\
&\quad\times
\sum_{k=0}^{m-1}\sum_{j=0}^{m-k-1}
\frac{(m-j-k)_j}{j!}
\left(\frac{r_{u,b}}{\tilde r_{u,b}}\right)^{\frac{j+k}{2}}
(\gamma+1)^k \gamma^{\frac{j-k}{2}}
I_{j+k}\!\left(
\frac{2\mu^2}{1-\mu^2}\,
\frac{\sqrt{\gamma r_{u,b}\tilde r_{u,b}}}{\gamma+1}
\right)
\end{aligned}
\end{equation}
\hrulefill
\end{figure*}

To reduce the computational complexity for \eqref{eq:FSIR_conditional_main}, we apply Gauss-Laguerre quadrature and derive the CDF as
\begin{equation}
\label{eq:FSIR_GLQ_cond_main}
\begin{aligned}
F_{\mathrm{SIR}_u}(\gamma\mid \bar A=m)
\approx
\prod_{b=1}^{B}&
\frac{1}{\Gamma(m)}
\sum_{\ell=1}^{N_r}
\sum_{\tilde\ell=1}^{N_{\tilde r}}
w_{\ell}\,\widetilde w_{\tilde\ell}^{(m)}
\\
&\times
\Bigl[
\Xi_m\bigl(\gamma;\zeta_{\ell},\widetilde\zeta_{\tilde\ell}^{(m)}\bigr)
\Bigr]^{L_b},
\end{aligned}
\end{equation}
where \(\{\zeta_\ell\}_{\ell=1}^{N_r}\) denote the \(N_r\) roots of the Laguerre polynomial \(L_{N_r}(x)\), and \(\{\widetilde\zeta_{\tilde\ell}^{(m)}\}_{\tilde\ell=1}^{N_{\tilde r}}\) denote the \(N_{\tilde r}\) roots of the generalized Laguerre polynomial \(L_{N_{\tilde r}}^{(m-1)}(x)\). For \(\ell=1,\ldots,N_r\), \(\tilde\ell=1,\ldots,N_{\tilde r}\), and \(m\in\{1,\ldots,U-1\}\), we have
\begin{align}
w_\ell
&\triangleq
\frac{\zeta_\ell}{(N_r+1)^2\bigl[L_{N_r+1}(\zeta_\ell)\bigr]^2},
\label{eq:GL_w_r}\\
\widetilde w_{\tilde\ell}^{(m)}
&\triangleq
\frac{\Gamma(N_{\tilde r}+m)}{N_{\tilde r}!\,(N_{\tilde r}+1)^2}\,
\frac{\widetilde\zeta_{\tilde\ell}^{(m)}}{\bigl[L_{N_{\tilde r}+1}^{(m-1)}(\widetilde\zeta_{\tilde\ell}^{(m)})\bigr]^2}.
\label{eq:GL_w_tilder}
\end{align}
When $\mu\to 1$, the block-diagonal approximation is precise, and the double-sum in \eqref{eq:Xi_main} becomes negligible. In this regime, the following approximation can be employed: 
\begin{equation}
\label{eq:Xi_mu_to_1_approx}
\begin{aligned}
\big[&\Xi_m\!\big(\gamma; r_{u,b},\tilde r_{u,b}\big)\big]^{L_b}\\
&\approx\!
\Biggl[
Q_m\!\left(
\frac{\mu}{\sqrt{1-\mu^{2}}}
\sqrt{\frac{2\gamma \tilde r_{u,b}}{\gamma+1}},
\frac{\mu}{\sqrt{1-\mu^{2}}}
\sqrt{\frac{2r_{u,b}}{\gamma+1}}
\right)
\Biggr]^{L_b}\!\!\!.
\end{aligned}
\end{equation}

By averaging the conditional CDF over the distribution of $\bar A$ for $m=1,\dots,U-1$, the CDF of the per-round SIR is obtained as
\begin{equation}\label{eq:FSIR_unconditional_main}
\begin{aligned}
F_{\mathrm{SIR}_u}(\gamma)
&=
\sum_{m=1}^{U-1}
\mathbb P(\bar A=m)\,
F_{\mathrm{SIR}_u}(\gamma\mid \bar A=m)
\\
&=
\sum_{m=1}^{U-1}
\frac{\binom{U-1}{m}p_{\mathrm a}^{m}(1-p_{\mathrm a})^{U-1-m}}
{1-(1-p_{\mathrm a})^{U-1}}
\\
&\quad\times
\prod_{b=1}^{B}
\int_{0}^{\infty}\!\!\int_{0}^{\infty}
\frac{\tilde r_{u,b}^{\,m-1}e^{-(r_{u,b}+\tilde r_{u,b})}}{\Gamma(m)}
\\
&\qquad~~~~\times
\bigl[\Xi_m(\gamma;r_{u,b},\tilde r_{u,b})\bigr]^{L_b}
\,dr_{u,b}\,d\tilde r_{u,b}.
\end{aligned}
\end{equation}
This CDF serves as the basis for the subsequent HARQ-CC outage analysis over multiple transmission rounds. 

\begin{theorem}[HARQ-CC cumulative-SIR characterization]
\label{thm:outage_conv_main}
Under the adopted block-fading model and the independent randomization of the interfering activity across HARQ rounds, for any $j\in\{1,\ldots,C\}$, the CDF of the SIR after $j$ HARQ rounds is given by the $j$-fold Stieltjes convolution
\begin{equation}
\label{eq:F_Gamma_j_main}
F_{\Gamma_u^{(q)}(j)}(\gamma)
=
F_{\mathrm{SIR}_u}^{\circledast j}(\gamma), 
~\gamma\ge 0,
\end{equation}
where $F_{\mathrm{SIR}_u}(\gamma)$ is the CDF of per-round SIR given in \eqref{eq:FSIR_unconditional_main}. Equivalently, $F_{\Gamma_u^{(q)}(j)}(\gamma)$ can be computed recursively by
\begin{subequations}
\label{eq:conv_recursion_main}
\begin{align}
F_{\Gamma_u^{(q)}(1)}(\gamma)&=F_{\mathrm{SIR}_u}(\gamma),\label{eq:conv_recursion_main_init}\\
F_{\Gamma_u^{(q)}(j)}(\gamma)&=\int_{0}^{\gamma}F_{\Gamma_u^{(q)}(j-1)}(\gamma-x)\,dF_{\mathrm{SIR}_u}(x),~2\le j\le C.
\label{eq:conv_recursion_main_step}
\end{align}
\end{subequations}
Consequently, the HARQ-CC outage probability at decoding threshold $\gamma_{\mathrm{th}}=2^{R_0}-1$ is given by
\begin{equation}\label{eq:Pout_conv_main}
P_{\mathrm{out}}
=
F_{\Gamma_u^{(q)}(C)}(\gamma_{\mathrm{th}})
=
F_{\mathrm{SIR}_u}^{\circledast C}(\gamma_{\mathrm{th}}). 
\end{equation}
\end{theorem}

\begin{proof}
The round-wise SIRs are i.i.d.\ with CDF in \eqref{eq:FSIR_unconditional_main}. Proposition~\ref{prop:sir_opt_combining} implies that the SIR after $C$ HARQ-CCs rounds is the sum of the $C$ round-wise SIRs. Therefore, the CDF of $\Gamma_u^{(q)}(C)$ is the $C$-fold Stieltjes convolution of per-round CDF, which yields \eqref{eq:F_Gamma_j_main} and \eqref{eq:conv_recursion_main}. Evaluating at $\gamma_{\mathrm{th}}$ gives \eqref{eq:Pout_conv_main}.
\end{proof}


\begin{remark}\label{rem:outage_building_block_main}
For numerical evaluation, the exact per-round CDF $F_{\mathrm{SIR}_u}(\gamma)$ in \eqref{eq:FSIR_unconditional_main} can be approximated using the Gauss-Laguerre formula. The resulting HARQ-CC outage probability can thus be efficiently computed by recursive convolution or, alternatively, via fast Fourier transform (FFT)-based discrete convolution on a suitable SIR grid. Additionally, Theorem~\ref{thm:outage_conv_main} provides the fundamental analytical component for the subsequent analysis of the average number of retransmissions, delay-related metrics, throughput, and energy efficiency.
\end{remark}

\vspace{-2mm}
\subsection{Average Number of HARQ Transmission Rounds}\label{subsec:avg_rounds}
For packet~$q$ of UE~$u$, let $\bar C\in\{1,2,\ldots,C\}$ denote the number of HARQ transmission rounds used by that packet, with $\bar C=C$ if decoding is not successful within the maximum of $C$ rounds. According to Theorem~\ref{thm:outage_conv_main}, the CDF of combined SIR, $F_{\Gamma_u^{(q)}(j)}$, is the $j$-fold Stieltjes convolution of $F_{\mathrm{SIR}_u}$. Also, $\Gamma_u^{(q)}(j)$ is continuous and nondecreasing in $j$ according to Proposition~\ref{prop:sir_opt_combining}. For simplicity, we write
$F_{\Gamma(j)}(\cdot)\equiv F_{\Gamma_u^{(q)}(j)}(\cdot)$. 

\begin{proposition}
\label{prop:avg_rounds_moments}
Under the HARQ stopping rule, the distribution of~$\bar C$ is given by
\begin{equation}
\label{eq:pmf_Cbar_1_new}
\mathbb P(\bar C=1)
=
\mathbb P\!\left(\Gamma_u^{(q)}(1)\ge\gamma_{\mathrm{th}}\right)
=
1-F_{\Gamma(1)}(\gamma_{\mathrm{th}}),
\end{equation}
for $i\in\{2,\ldots,C-1\}$,
\begin{equation}
\label{eq:pmf_Cbar_i_new}
\begin{aligned}
\mathbb P(\bar C=i)
&=
\mathbb P\!\left(
\Gamma_u^{(q)}(i-1)<\gamma_{\mathrm{th}},\,
\Gamma_u^{(q)}(i)\ge\gamma_{\mathrm{th}}
\right) \\
&=
F_{\Gamma(i-1)}(\gamma_{\mathrm{th}})
-
F_{\Gamma(i)}(\gamma_{\mathrm{th}}),
\end{aligned}
\end{equation}
and
\begin{equation}
\label{eq:pmf_Cbar_C_new}
\mathbb P(\bar C=C)
=
\mathbb P\!\left(\Gamma_u^{(q)}(C-1)<\gamma_{\mathrm{th}}\right)
=
F_{\Gamma(C-1)}(\gamma_{\mathrm{th}}).
\end{equation}
Accordingly, the first two moments of~$\bar C$ admit the tail-sum representations
\begin{equation}
\label{eq:ECbar_tail_main_new}
\mathbb E[\bar C]
=
\sum_{i=1}^{C}\mathbb P(\bar C\ge i)
=
1+\sum_{j=1}^{C-1}F_{\Gamma(j)}(\gamma_{\mathrm{th}}),
\end{equation}
and
\begin{equation}
\label{eq:EC2_tail_main_new}
\mathbb E[\bar C^2]
=
\sum_{i=1}^{C}(2i-1)\,\mathbb P(\bar C\ge i)
=
1+\sum_{j=1}^{C-1}(2j+1)\,F_{\Gamma(j)}(\gamma_{\mathrm{th}}).
\end{equation}
\end{proposition}

\begin{proof}
See Appendix~\ref{app:proof_avg_transmissions}.
\end{proof}

In the subsequent analysis, we use the tail-sum forms in \eqref{eq:ECbar_tail_main_new} and \eqref{eq:EC2_tail_main_new}, since these forms are numerically more robust than those obtained by the moments via \eqref{eq:pmf_Cbar_1_new}--\eqref{eq:pmf_Cbar_C_new}, involving subtraction of adjacent CDF values.

\vspace{-2mm}
\subsection{Average Waiting Time}\label{subsec:avg_waiting_time} 
We focus on a tagged UE~$u$ operating a stop-and-wait HARQ process, allowing at most one active HOL packet at any time. Time is slotted into fading blocks, and under the assumption of block-synchronous HARQ, each round precisely occupies one fading block of duration $T_F$ seconds. Consequently, a new packet can start service only at a block boundary. For packet~$q$, let $S_q$ denote its service time in seconds, representing the duration from the moment the packet begins service until it is either successfully decoded or declared in outage. Given that at most $C$ HARQ rounds are permitted and each round occupies one fading block of duration $T_F$, the service time of packet~$q$ can be written as
\begin{equation}
S_q=\bar C\,T_F,
\end{equation}
where $\bar C\in\{1,\ldots,C\}$ denotes the realized number of HARQ rounds. The first two service-time moments are given by
\begin{equation}
\left\{\begin{aligned}
s_1 &\triangleq \mathbb E[S_q]=T_F\,\mathbb E[\bar C],\\
s_2 &\triangleq \mathbb E[S_q^2]=T_F^2\,\mathbb E[\bar C^2],
\end{aligned}\right.
\end{equation}
where the moments $\mathbb E[\bar C]$ and $\mathbb E[\bar C^2]$ are given in \eqref{eq:ECbar_tail_main_new} and \eqref{eq:EC2_tail_main_new}, respectively. These moments depend on the block-level interferer activity probability $p_{\mathrm a}$, considering the SIR distributions $\{F_{\Gamma(j)}\}_{j=1}^{C}$ derived in Section~\ref{subsec:outage_analysis}

In the underlying symmetric interacting system, $p_{\mathrm a}$ is coupled with the steady-state queue busy fraction through the mutual interaction between interference and service dynamics. Under the mean-field approximation, this coupling is characterized by the fixed-point relation $p_{\mathrm a}=\rho$. Consequently, $\mathbb E[\bar C]$ and $\mathbb E[\bar C^2]$, as well as the resulting system-level metrics, are determined through a mean-field fixed-point closure.

Assume that packet arrivals to each UE form a Poisson process with rate $\lambda$ packets/s. For a fixed value of $p_{\mathrm a}$, the tagged queue is approximated under the homogeneity and mean-field decoupling assumptions as an $M/G/1$ queue with i.i.d.\ service times $\{S_q\}$. Its utilization is therefore given by
\begin{equation}\label{eq:rho_def_pa}
\rho(p_{\mathrm a})
\triangleq
\lambda\,\mathbb E[S_q]
=
\lambda s_1(p_{\mathrm a}),
\end{equation}
and its stability requires $\rho(p_{\mathrm a})<1$.


Under the $M/G/1$ approximation, the tagged queue alternates between busy and idle periods. Hence, at steady state, $\rho(p_{\mathrm a})$ represents the long-term busy fraction of the queue \cite{Pollaczek}. In the present grant-free, work-conserving, block-synchronous setting, a non-tagged UE is active in a fading block if and only if its queue is busy serving its HOL packet in that block. Under the mean-field approximation, the typical-interferer activity probability therefore coincides with the common steady-state busy fraction, which yields the closure
\begin{equation}\label{eq:pa_fixed_point_logic}
p_{\mathrm a}
=
\rho(p_{\mathrm a})
=
\lambda T_F\,\mathbb E[\bar C \mid p_{\mathrm a},\gamma_{\mathrm{th}}],
\end{equation}
where the dependence of $\mathbb E[\bar C]$ on $p_{\mathrm a}$ is induced by the interference model in Section~\ref{subsec:outage_analysis}. When adopting the conditioned interferer-count model, the conditioning affects the evaluation of $\mathbb E[\bar C]$ and $\mathbb E[\bar C^2]$ through the per-round SIR law, but it does not alter the queue-level closure in \eqref{eq:pa_fixed_point_logic}.
To formalize this self-consistency relation, we define
\begin{equation}\label{eq:mf_map}
f(p_{\mathrm a};\gamma_{\mathrm{th}})\triangleq\lambda T_F\,\mathbb E[\bar C \mid p_{\mathrm a},\gamma_{\mathrm{th}}],~p_{\mathrm a}\in[0,1].
\end{equation}

A self-consistent steady-state operating point is any fixed point $p_{\mathrm a}^\star\in[0,1]$ satisfying
$
p_{\mathrm a}^\star=f(p_{\mathrm a}^\star;\gamma_{\mathrm{th}}).
$
Queue stability further requires $p_{\mathrm a}^\star<1$, equivalently $\rho(p_{\mathrm a}^\star)<1$. In the numerical results, $p_{\mathrm a}^\star$ is computed for each parameter configuration as a queue-stable solution of the above fixed-point equation, whenever such a solution exists.
The HARQ limit $C$ affects the self-consistent operating point only through its impact on the mean service time $\mathbb E[\bar C]$. Since $\bar C\le C$ almost surely, the mean-field map satisfies
\[
f(p_{\mathrm a};\gamma_{\mathrm{th}})
\le
\lambda T_F C,
~\forall\, p_{\mathrm a}\in[0,1].
\]
Hence, if $\lambda T_F C<1$, then any fixed point of \eqref{eq:pa_fixed_point_logic}, whenever it exists, is automatically queue-stable.

By the Pollaczek-Khinchin formula \cite{Pollaczek}, the mean queueing waiting time is given by
\begin{equation}\label{eq:PK_wait}
\mathbb E[W_q]=\frac{T_F}{2}+\frac{\lambda\,\mathbb E[S_q^2]}{2(1-\rho)}=\frac{\lambda\,T_F^2\,\mathbb E[\bar C^2]}{2(1-\rho)}
+\frac{T_F}{2},
\end{equation}
where $W_q$ denotes the elapsed time from the arrival epoch of packet~$q$ until the block boundary at which it starts service, and the term $T_F/2$ is the mean block-alignment delay.

Thus, the mean sojourn time of packet~$q$ is given by
\begin{equation}\label{eq:sojourn}
\mathbb E[T_q]=\mathbb E[S_q]+\mathbb E[W_q]=T_F\,\mathbb E[\bar C]+\frac{\lambda\,T_F^2\,\mathbb E[\bar C^2]}{2(1-\rho)}+\frac{T_F}{2}.
\end{equation}
Substituting \eqref{eq:ECbar_tail_main_new} and \eqref{eq:EC2_tail_main_new} into \eqref{eq:PK_wait} and \eqref{eq:sojourn} yields the waiting-time and sojourn-time metrics as explicit functions of the HARQ-CC cumulative-SIR distributions $\{F_{\Gamma(j)}\}_{j=1}^{C}$.


\vspace{-2mm}
\subsection{Payload Throughput and Energy Efficiency} 
We use resource-normalized payload metrics, under which performance is evaluated based on the allocated service resources.
The per-UE payload throughput is defined as the number of payload bits successfully transmitted per occupied channel use, i.e.,
\begin{equation}\label{eq:throughput_resource_pay}
\mathcal T_u^{(\mathrm{pay})}\triangleq\frac{M_s(1-P_{\mathrm{out}})}{L\,\mathbb E[\bar C]}=\frac{R_0(1-P_{\mathrm{out}})}{\mathbb E[\bar C]}.
\end{equation}
With homogeneous user setting, the aggregate system payload throughput is obtained as $U\,\mathcal T_u^{(\mathrm{pay})}$.

Furthermore, the payload energy efficiency is defined as the ratio of the expected payload bits successfully transmitted per packet to the average transmit energy consumed in servicing that packet, i.e.,
\begin{equation}\label{eq:ee_pay_resource_def}
\eta_u^{(\mathrm{pay})}\triangleq\frac{M_s(1-P_{\mathrm{out}})}{L E_s\,\mathbb E[\bar C]}~\text{(bits/J)}.
\end{equation}
Equivalently, the relationship between the payload throughput and the energy efficiency can be expressed as
\begin{equation}\label{eq:ee_pay_resource_compact}
\eta_u^{(\mathrm{pay})}=\frac{\mathcal T_u^{(\mathrm{pay})}}{E_s}.
\end{equation}
Both metrics are active-link quantities and therefore depend on $P_{\mathrm{out}}$, $\mathbb E[\bar C]$, and $E_s$. 

\vspace{-2mm}
\section{Simulation Results}\label{Simulation Results}
In this section, numerical and Monte Carlo (MC) results are provided for the considered HARQ-\emph{s}FAMA system under the queue-coupled mean-field model. The system consists of $8$ UEs sharing the same time--frequency resource. To ensure a physically meaningful spatial model, the 1D-FAS aperture at each UE is set to \(15\) cm and the carrier frequency is fixed at \(7\) GHz, which corresponds to a normalized size of \(W=3.50\). The MC simulations are based on the full spatial correlated Jakes model, while the analytical results are approximated by the block-correlation model. The block sizes of the approximated model are computed by \cite[Algorithm~1]{spatialblockmodel} using the eigenvalue threshold \(\rho_{\mathrm{th}}=1\), and the effective parameter \(\mu\) is chosen to match the nearest-neighbor correlation of the full model. 
The analytical results are evaluated using the Gauss-Laguerre quadrature rules in \eqref{eq:GL_w_r} and \eqref{eq:GL_w_tilder}, with quadrature orders set to \((N_r,N_{\tilde r})=(30,30)\). The frame duration is set to \(T_F=1\) ms, and the average symbol energy is normalized to \(E_s=1\). As a benchmark, a conventional FPA receiver is considered. Moreover, the target transmission rate is chosen consistently with the SIR threshold as \(R_0=\log_2(1+\gamma_{\mathrm{th}})\).

\begin{figure}[tbp]
\centering
\includegraphics[width=.95\linewidth]{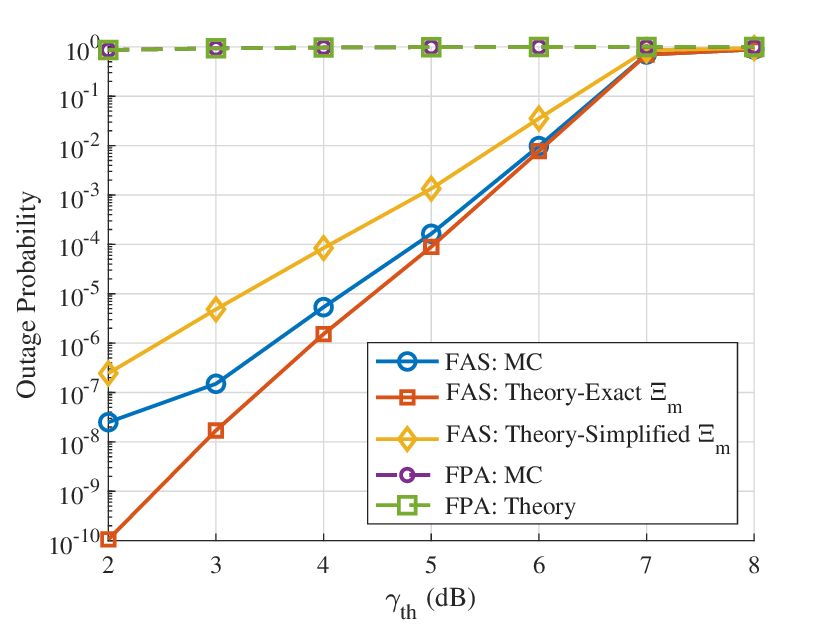}
\caption{Outage probability of HARQ-\emph{s}FAMA versus the SIR threshold $\gamma_{\mathrm{th}}$, comparing the FAS scheme and the FPA baseline.}\label{fig:lambda_outage}
\vspace{-2mm}
\end{figure}

\begin{figure}[tbp]
\begin{center}
\subfigure[Average waiting time]{\includegraphics[width = .95\linewidth]{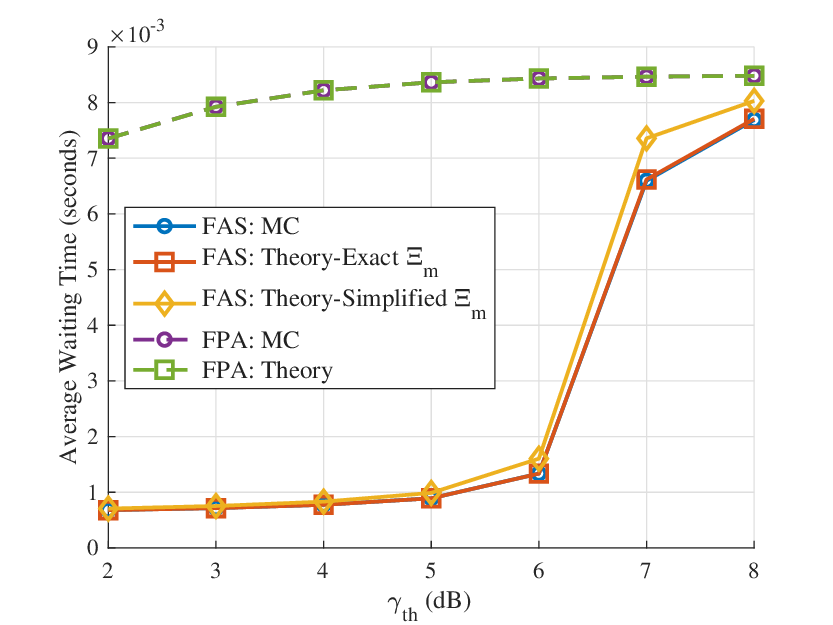}}\\
\vspace{-2mm}
\subfigure[Energy efficiency]{\includegraphics[width = .95\linewidth]{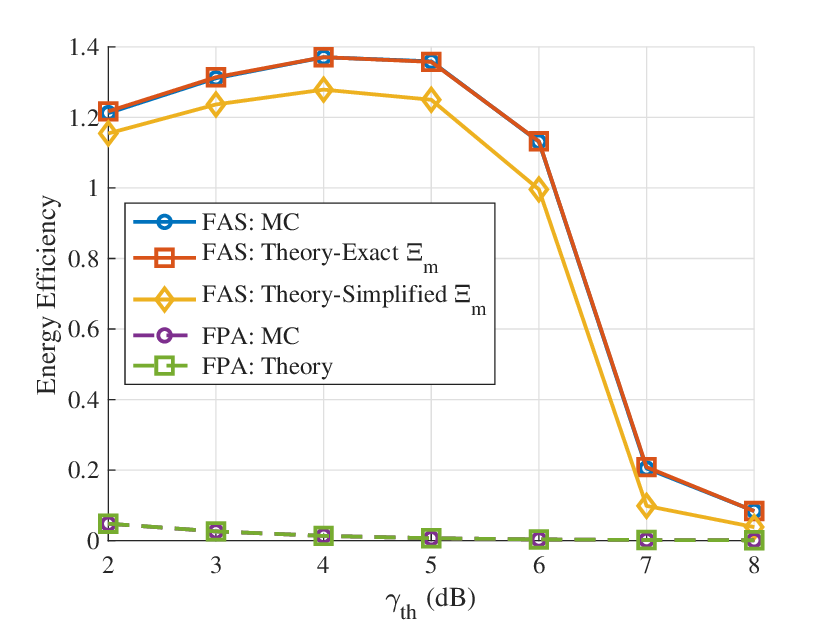}}\\
\vspace{-2mm}
\subfigure[Busy fraction]{\includegraphics[width = .95\linewidth]{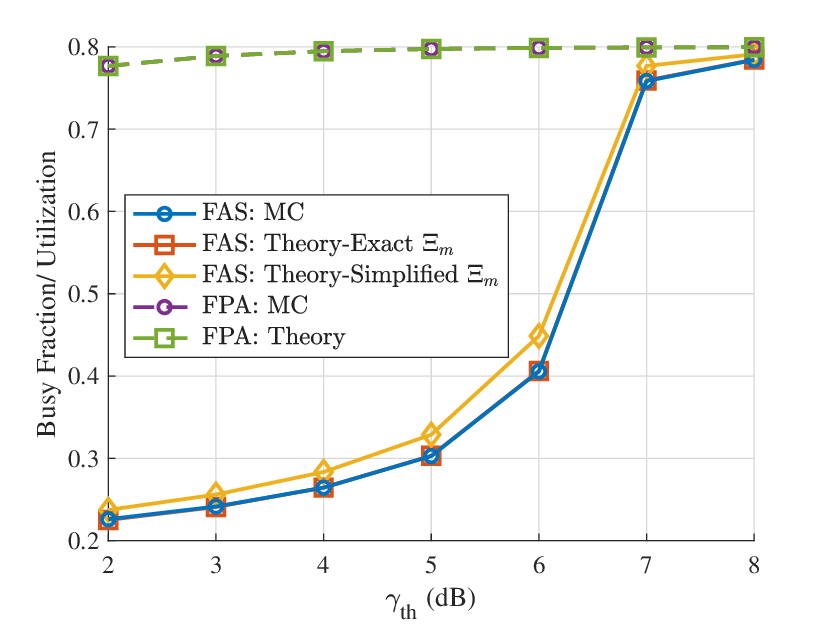}}
\caption{System-level performance of HARQ-\emph{s}FAMA versus the SIR threshold $\gamma_{\mathrm{th}}$: (a) average waiting time, (b) energy efficiency, and (c) busy fraction.}\label{fig:lambda_system}
\vspace{-2mm}
\end{center}
\end{figure}

Fig.~\ref{fig:lambda_outage} shows the outage probability as a function of the SIR threshold \(\gamma_{\mathrm{th}}\). It is seen that the exact analytical results closely match the MC simulation results for both the FAS and FPA schemes, thereby validating the accuracy of the analytical framework. The simplified approximation also captures the overall trend well, while remaining slightly conservative. The accuracy of the approximation improves as \(\mu\) approaches to $1$. Moreover, HARQ-\emph{s}FAMA with FAS receiver consistently achieves a lower outage probability than the FPA baseline.

In Fig.~\ref{fig:lambda_system}, we present the system-level performance, including average waiting time, energy efficiency, and busy fraction. For small values of $\gamma_{\mathrm{th}}$, the decoding requirement is modest, thereby ensuring that the target threshold is surpassed with high probability after only one or a few HARQ rounds. This results in a short average service time, which consequently leads to a small busy fraction and less waiting times, and as a result improves energy efficiency. On the other hand, as $\gamma_{\mathrm{th}}$ increases, the average number of HARQ rounds grows, which increases the service time and causes the utilization factor $\rho$ to approach $1$. As a consequence, the system approaches its stability limit, resulting in a rapid increase in the busy fraction and waiting time, while the energy efficiency experiences a significant decline. Despite this, compared with FPA, the FAS scheme demonstrates much lower outage probability, much lower delay and greater utilization, and much enhanced energy efficiency across the entire range of threshold.

\begin{figure*}[tbp]
\centering
\subfigure[Outage probability]{
\includegraphics[width=.49\textwidth]{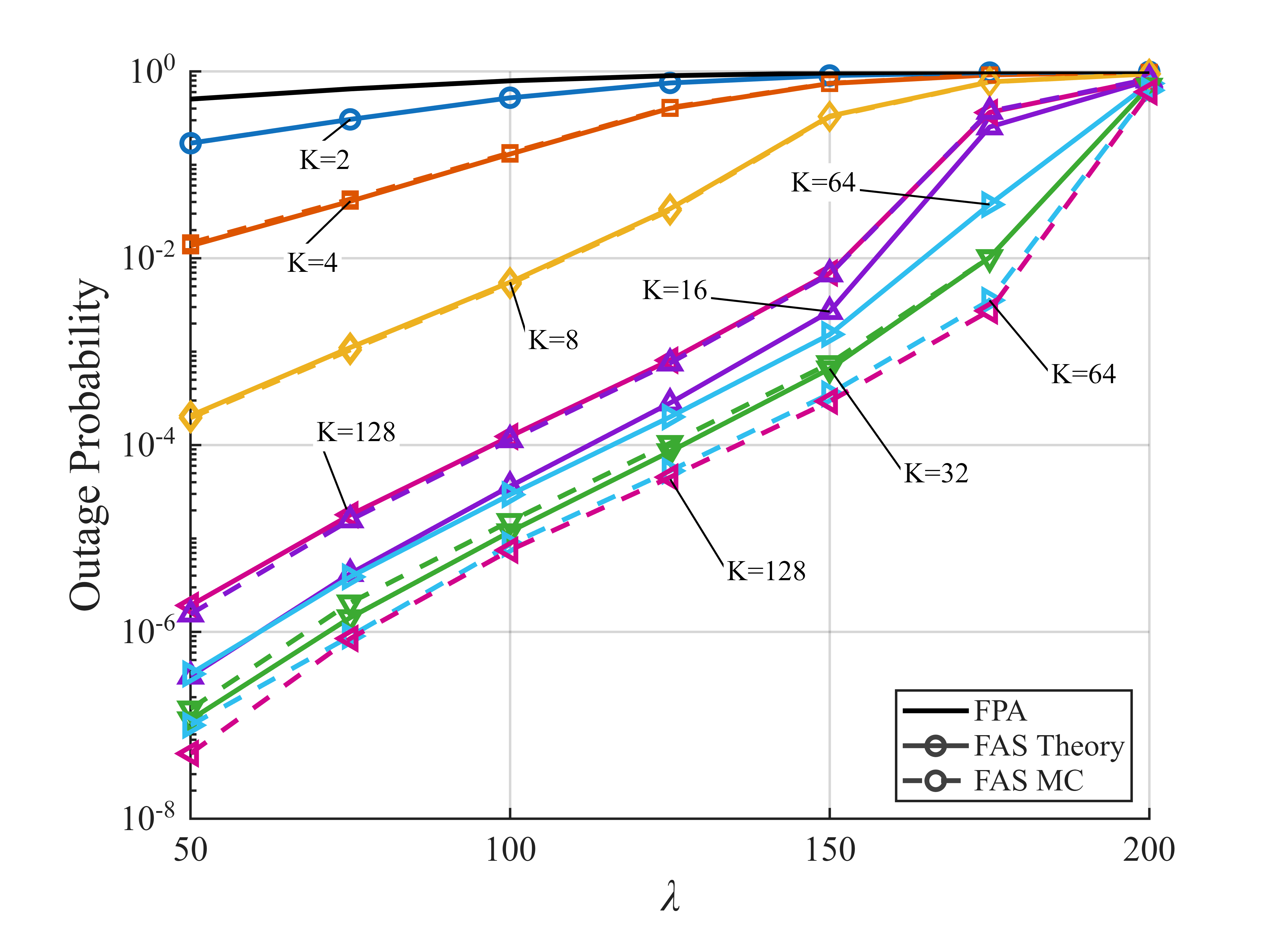}}
\hfill
\subfigure[Average waiting time]{
\includegraphics[width=.49\textwidth]{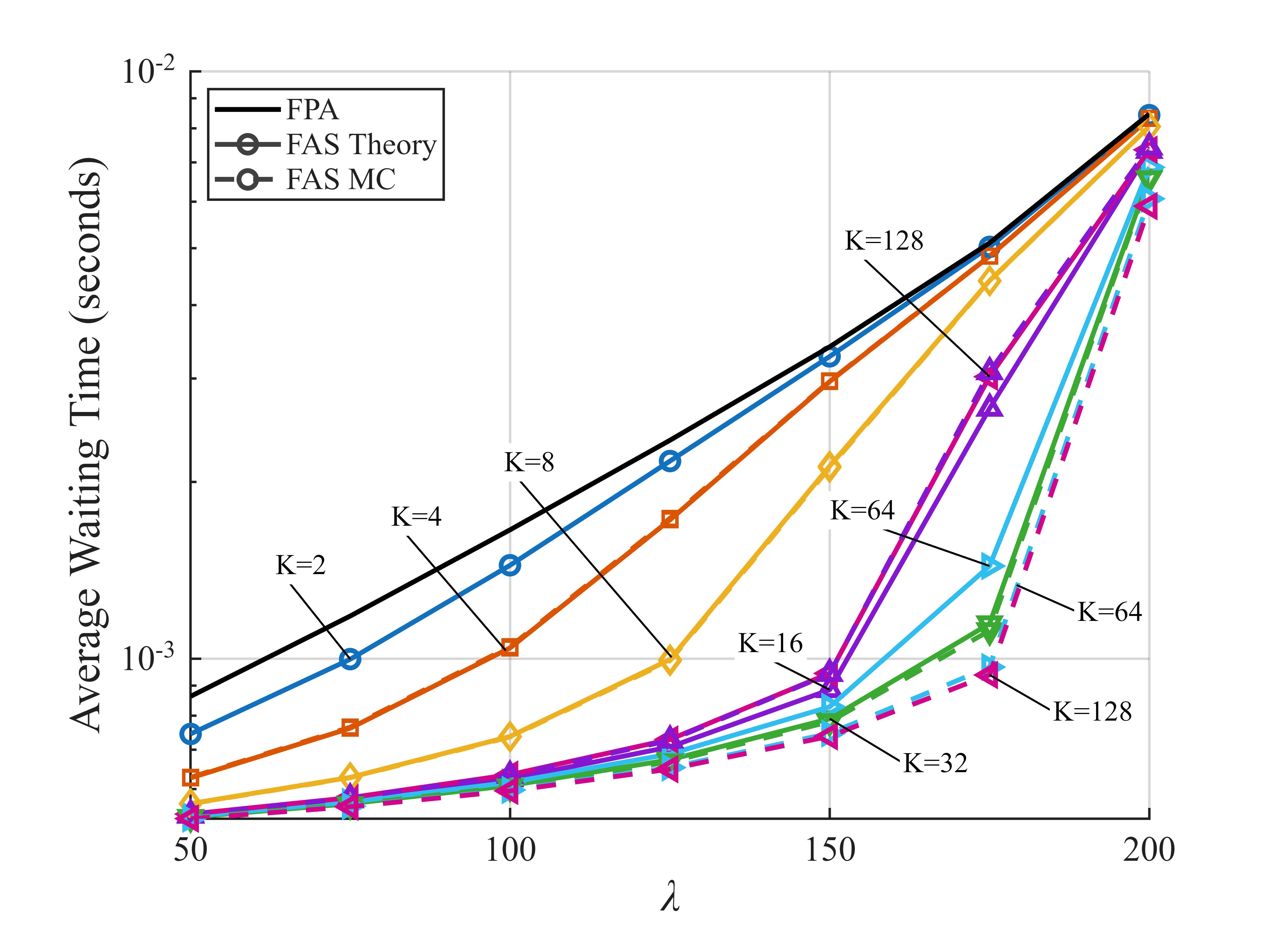}}
\vspace{-2mm}
\subfigure[Energy efficiency]{
\includegraphics[width=.49\textwidth]{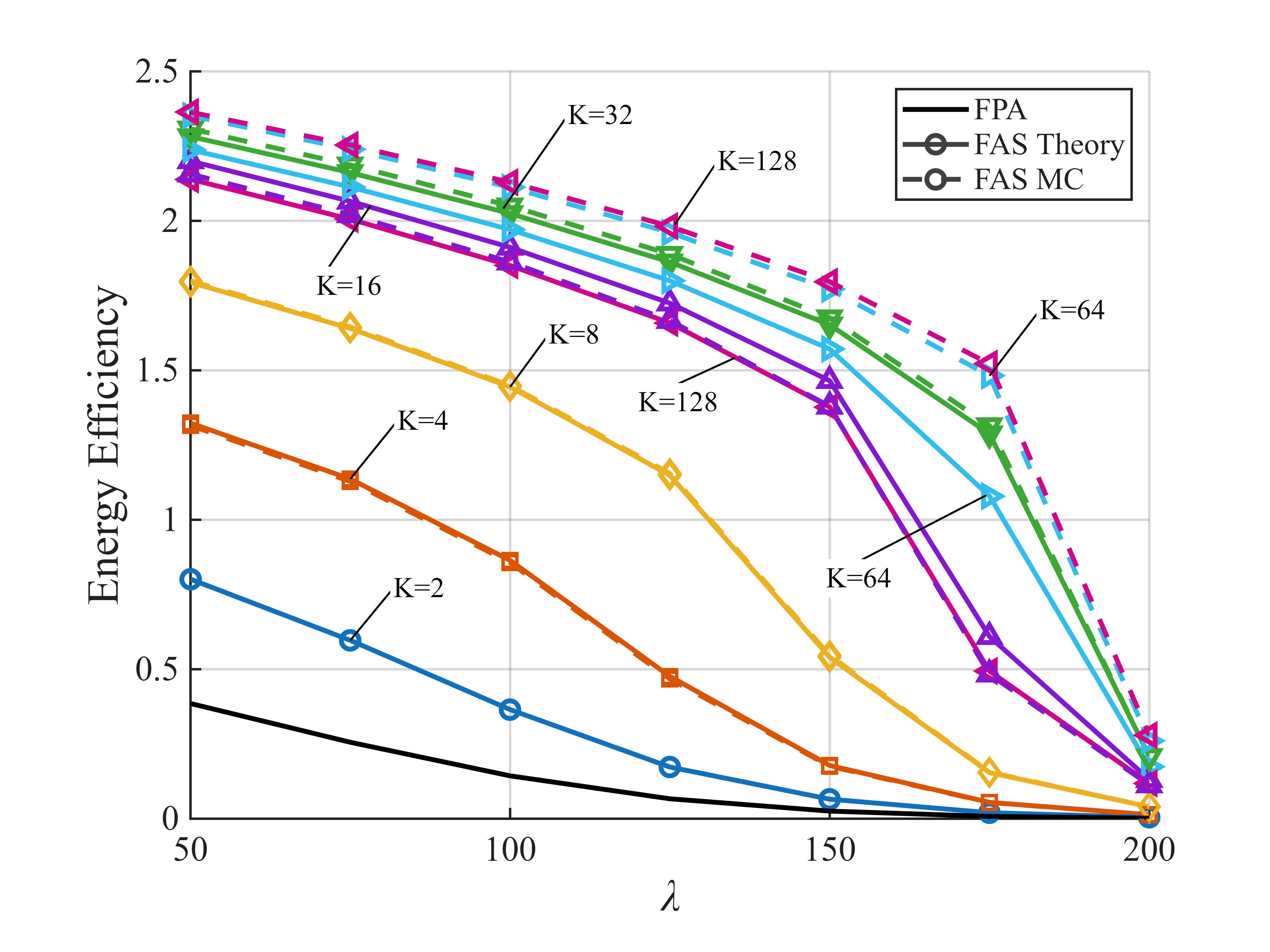}}
\hfill
\subfigure[Busy fraction]{
\includegraphics[width=.49\textwidth]{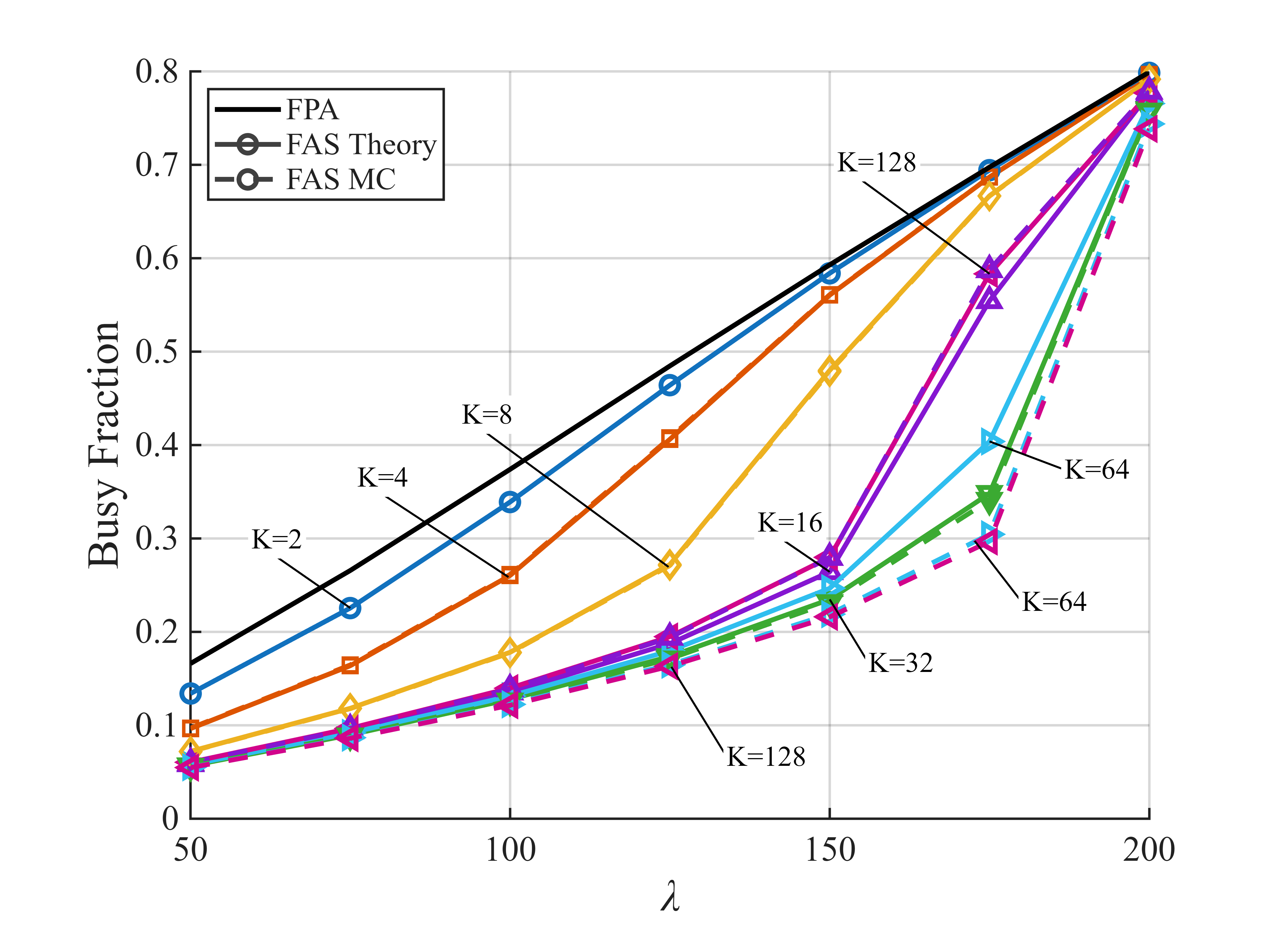}}
\caption{Outage probability and system-level performance of HARQ-\emph{s}FAMA versus the arrival rate $\lambda$ at $\gamma_{\rm th}=7~{\rm dB}$, with different $K$ for $C=4$.}
\label{fig:sys_lambdaK}
\vspace{-2mm}
\end{figure*}

\begin{figure*}[tbp]
\centering
\subfigure[Outage probability]{
\includegraphics[width=.46\textwidth]{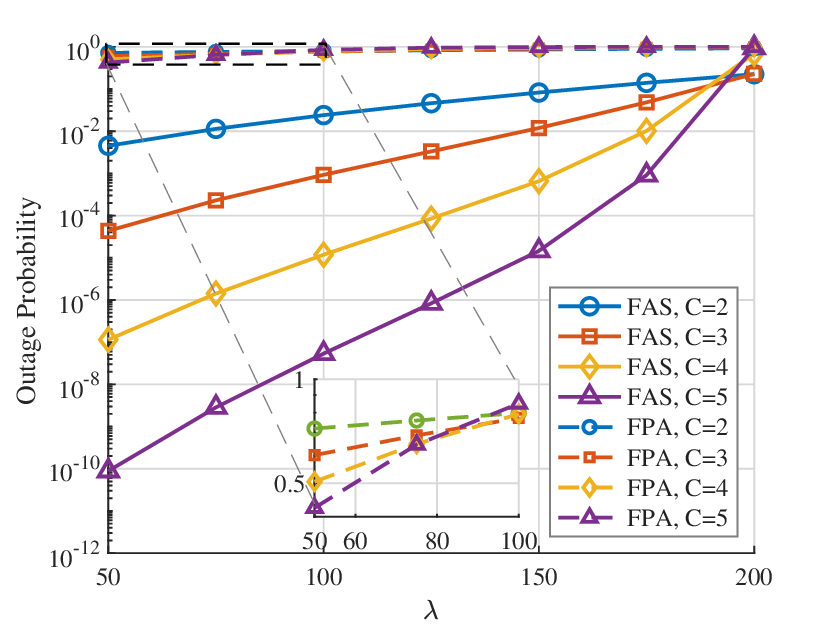}}
\hfill
\subfigure[Average waiting time]{
\includegraphics[width=.46\textwidth]{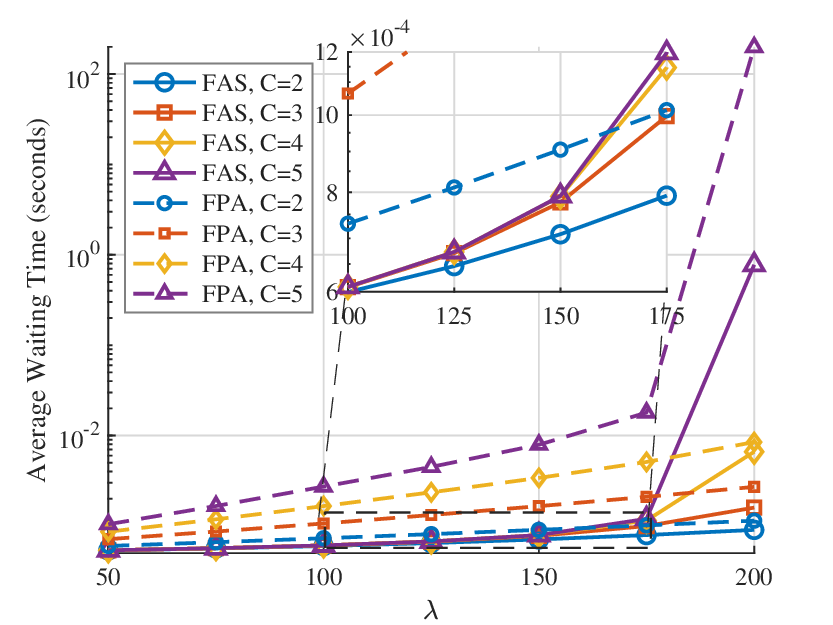}}
\vspace{-2mm}
\subfigure[Energy efficiency]{
\includegraphics[width=.46\textwidth]{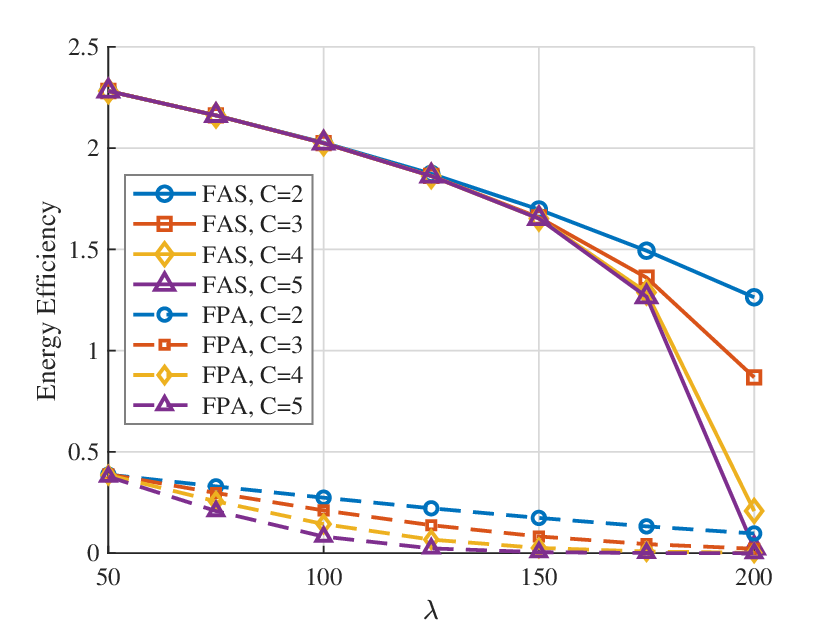}}
\hfill
\subfigure[Busy fraction]{
\includegraphics[width=.46\textwidth]{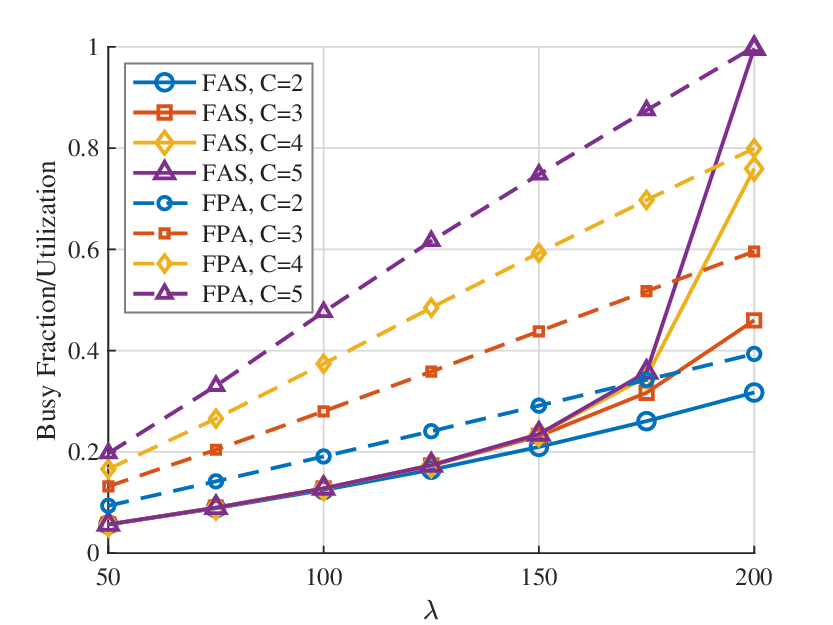}}
\caption{Outage probability and system-level performance of HARQ-\emph{s}FAMA versus the arrival rate $\lambda$ at $\gamma_{\rm th}=7~{\rm dB}$, with different $C$ for $K=32$.}
\label{fig:sys_lambdaC}
\vspace{-2mm}
\end{figure*}

The results in Figs.~\ref{fig:lambda_outage} and \ref{fig:lambda_system} illustrate that HARQ-\emph{s}FAMA experiences a pronounced performance degradation as the SIR threshold approaches \(7\) dB, indicating that \(\gamma_{\mathrm{th}}=7\) dB is a critical operating point near the threshold of queue saturation. We thus fix \(\gamma_{\mathrm{th}}=7\) dB for subsequent analysis and study how the performance at this threshold can be improved through two key design parameters: the maximum HARQ round limit \(C\) and the number of FAS ports \(K\). These parameters improve reliability via distinct mechanisms: increasing \(C\) allows for additional retransmission and combining opportunities, whereas increasing \(K\) enhances spatial selection diversity by offering more candidate receive antenna ports. Unless otherwise specified, the arrival rate \(\lambda\) is varied over a range that preserves queue stability for all considered parameter configurations. Specifically, as per \eqref{eq:rho_def_pa} with \(T_F=1\) ms and \(\mathbb{E}[\bar C]\le C\), operating points approaching \(\rho=1\) are excluded. 

We analyze the impact of \(K\) on both outage and system-level performance in Fig.~\ref{fig:sys_lambdaK}. With \(C=4\), the results show that increasing \(K\) is beneficial when \(K\leq 32\). In this regime, the analytical results closely match the MC results, and a larger candidate-port set provides more effective selection diversity, leading to a lower outage probability, reduced average waiting time and busy fraction, and improved energy efficiency. For small candidate-port sets, e.g., \(K=2,4,8\), the busy fraction and average waiting time increase rapidly as \(\lambda\) grows, while the energy efficiency decreases sharply. This indicates that the system is more sensitive to the traffic load when the available selection diversity is limited. In contrast, for larger candidate-port sets, e.g., \(K=32,64,128\), the busy fraction and average waiting time grow more slowly, and the energy efficiency decreases more gradually, showing improved robustness against the increase in traffic load. Moreover, the FPA baseline exhibits similar qualitative trends, whereas FAS consistently achieves better system-level performance under the same conditions.

When \(K>32\), the additional ports become more correlated under the fixed FAS size and thus provide only limited new effective diversity. Therefore, the MC performance saturates when \(K\) is large. However, we observe a slight performance degradation predicted by the analytical curves for \(K>32\). This degradation is due to the conservatism introduced by the simplified approximation. This conservatism becomes more pronounced in the dense-port regime, where the spatial block-correlation approximation becomes more sensitive to its mismatch with the actual \(J_0(\cdot)\)-based correlation structure. The resulting approximation error is magnified by the \(L_b\)-fold exponentiation within each block and the product operation across blocks, and is further propagated through the queue-level fixed-point closure in \eqref{eq:pa_fixed_point_logic}. Consequently, the analytical model tends to overestimate the average number of transmissions, which yields a larger fixed point for \(p_{\rm a}\) and hence a higher predicted interference level. This leads to conservative outage and system-level performance predictions, as evidenced by the larger analytical busy fraction in Fig.~\ref{fig:sys_lambdaK}(d).

Fig.~\ref{fig:sys_lambdaC} investigates the outage probability and system-level performance versus the arrival rate \(\lambda\) at \(\gamma_{\rm th}=7~{\rm dB}\), with different maximum HARQ round limits \(C\) and fixed \(K=32\). As shown in Fig.~\ref{fig:sys_lambdaC}(a), increasing \(C\) significantly reduces the outage probability for a given \(\lambda\) by providing more HARQ combining opportunities before packet failure declaration. The corresponding system-level results are shown in Figs.~\ref{fig:sys_lambdaC}(b)--(d). The impact of the maximum HARQ round limit \(C\) is limited when $\lambda$ is low, where the busy fraction, average waiting time, and energy efficiency are all close for different \(C\). As \(\lambda\) increases, however, the effect of \(C\) becomes more pronounced. In particular, when \(\lambda\) exceeds approximately \(175\), a larger \(C\) tends to increase the busy fraction and waiting time while reducing the energy efficiency. This is because although a larger \(C\) improves reliability, it also allows more retransmission attempts per packet, thereby increasing the service burden under heavy traffic. Therefore, a larger HARQ round budget does not necessarily improve queueing performance at high offered load. Across all performance metrics, FAS consistently outperforms the FPA baseline over the entire range of \(\lambda\).

\vspace{-2mm}
\section{Conclusion}\label{Conclusion}
This paper considered a downlink HARQ-\emph{s}FAMA system, where each UE with a FAS performs port selection in each HARQ round, utilizing a designed signal combination across multiple rounds. Furthermore, we provided a tractable analytical framework for its performance characterization over block-fading channels. Specifically, the outage probability, average packet waiting time, and energy efficiency were analyzed via an $M/G/1$ approximation with mean-field closure. Numerical results verified the accuracy of the proposed framework for moderate numbers of ports and showed that HARQ-\emph{s}FAMA consistently outperforms its HARQ-assisted FPA benchmark in terms of reliability, delay, and energy efficiency. More importantly, the results demonstrated that performance saturates when the number of ports increases, suggesting that moderate HARQ and port-density settings can more effectively balance reliability, delay, and energy efficiency.

\vspace{-2mm}
\appendices
\section{Proof of Proposition~\ref{prop:sir_opt_combining}}\label{app:proof_prop_sir_opt}
To determine the SIR-optimal combining vector, we maximize the combined SIR in \eqref{eq:post_SIR_vec_main} over all nonzero $\boldsymbol{\alpha}\in\mathbb C^i$, which leads to
\begin{equation}\label{prob:maxSIR_alpha_app}
\max_{\boldsymbol{\alpha}\in\mathbb C^{i}}\frac{\big|\boldsymbol{\alpha}^{\mathsf H}\mathbf g_{u}^{(q)}(i)\big|^{2}}
{\sigma^2(1-\mu^2)\,\boldsymbol{\alpha}^{\mathsf H}\mathbf D_{u}^{(q)}(i)\,\boldsymbol{\alpha}}~~\text{s.t.}~~\boldsymbol{\alpha}\neq \mathbf 0.
\end{equation}
Since the objective is invariant to any nonzero scalar scaling of $\boldsymbol{\alpha}$, we may impose the normalization $\boldsymbol{\alpha}^{\mathsf H}\mathbf D_{u}^{(q)}(i)\,\boldsymbol{\alpha}=1$ without loss of optimality. Hence, \eqref{prob:maxSIR_alpha_app} is equivalent to
\begin{subequations}\label{prob:maxSIR_alpha_norm_app}
\begin{align}
\max_{\boldsymbol{\alpha}\in\mathbb C^{i}}~~&
\big|\boldsymbol{\alpha}^{\mathsf H}\mathbf g_{u}^{(q)}(i)\big|^{2}
\label{prob:maxSIR_alpha_norm_app_obj}\\
\text{s.t.}~~&
\boldsymbol{\alpha}^{\mathsf H}\mathbf D_{u}^{(q)}(i)\,\boldsymbol{\alpha}=1 .
\label{prob:maxSIR_alpha_norm_app_con}
\end{align}
\end{subequations}
Let $\big(\mathbf D_{u}^{(q)}(i)\big)^{1/2}$ denote the principal square root of $\mathbf D_{u}^{(q)}(i)$, and define
\begin{equation}\label{eq:whitening_app}
\left\{\begin{aligned}
\boldsymbol{\beta}&\triangleq \big(\mathbf D_{u}^{(q)}(i)\big)^{1/2}\boldsymbol{\alpha},\\
\boldsymbol{\xi}&\triangleq \big(\mathbf D_{u}^{(q)}(i)\big)^{-1/2}\mathbf g_{u}^{(q)}(i).
\end{aligned}\right.
\end{equation}
Then \eqref{prob:maxSIR_alpha_norm_app_con} is rewritten as
\begin{equation}\label{eq:beta_constraint_app}
\|\boldsymbol{\beta}\|^2=1,
\end{equation}
and the objective can be derived as
\begin{equation}\label{eq:beta_objective_app}
\big|\boldsymbol{\alpha}^{\mathsf H}\mathbf g_{u}^{(q)}(i)\big|^2=\big|\big(\big(\mathbf D_{u}^{(q)}(i)\big)^{-1/2}\boldsymbol{\beta}\big)^{\mathsf H}\mathbf g_{u}^{(q)}(i)\big|^2=\big|\boldsymbol{\beta}^{\mathsf H}\boldsymbol{\xi}\big|^2.
\end{equation}
Thus, the optimization problem can be represented as
\begin{equation}\label{prob:maxSIR_beta_app}
\max_{\boldsymbol{\beta}\in\mathbb C^{i}}\big|\boldsymbol{\beta}^{\mathsf H}\boldsymbol{\xi}\big|^2~~\text{s.t.}~~\|\boldsymbol{\beta}\|^2=1.
\end{equation}
By the Cauchy-Schwarz inequality, we have
\begin{equation}\label{eq:CS_app}
\big|\boldsymbol{\beta}^{\mathsf H}\boldsymbol{\xi}\big|\le\|\boldsymbol{\beta}\|\,\|\boldsymbol{\xi}\|=\|\boldsymbol{\xi}\|,
\end{equation}
with equality if and only if $\boldsymbol{\beta}$ is colinear with $\boldsymbol{\xi}$. Hence,
\begin{equation}\label{eq:beta_opt_app}
\boldsymbol{\beta}^{\star}=e^{\mathrm j\phi}\,\frac{\boldsymbol{\xi}}{\|\boldsymbol{\xi}\|},~\phi\in[0,2\pi),
\end{equation}
and therefore
\begin{equation}\label{eq:alpha_opt_app}
\boldsymbol{\alpha}^{\star}=\big(\mathbf D_{u}^{(q)}(i)\big)^{-1/2}\boldsymbol{\beta}^{\star}=e^{\mathrm j\phi}\,
\frac{\big(\mathbf D_{u}^{(q)}(i)\big)^{-1}\mathbf g_{u}^{(q)}(i)}{\sqrt{\mathbf g_{u}^{(q)}(i)^{\mathsf H}\big(\mathbf D_{u}^{(q)}(i)\big)^{-1}\mathbf g_{u}^{(q)}(i)}}.
\end{equation}
This proves \eqref{eq:alpha_opt_vec_main}, with the arbitrary global factor absorbed into $c_i\in\mathbb C\setminus\{0\}$. Substituting \eqref{eq:alpha_opt_app} into \eqref{eq:post_SIR_vec_main} yields
\begin{equation}\label{eq:SIR_opt_app}
\mathrm{SIR}_{u}^{(q)}\!\left(i;\boldsymbol{\alpha}_{u}^{\star}\right) 
=\frac{\big(\mathbf g_{u}^{(q)}(i)\big)^{\mathsf H}\big(\mathbf D_{u}^{(q)}(i)\big)^{-1}\mathbf g_{u}^{(q)}(i)}{\sigma^2(1-\mu^2)}.
\end{equation}
Finally, using \eqref{eq:h_D_def_main} and \(|g_{u,u,k}^{(j)}|^2=\sigma^2(1-\mu^2)S_{u,k}^{(j)}\), we get
\begin{align}
\big(\mathbf g&_{u}^{(q)}(i)\big)^{\mathsf H}\big(\mathbf D_{u}^{(q)}(i)\big)^{-1}\mathbf g_{u}^{(q)}(i) \nonumber\\
&=\sum_{j=1}^{i}\frac{\big|g_{(u,u),k_u^{(j)}}^{(j)}\big|^2}{I_{u,k_u^{(j)}}^{(j)}}
=\sigma^2(1-\mu^2)\sum_{j=1}^{i}\frac{S_{u,k_u^{(j)}}^{(j)}}{I_{u,k_u^{(j)}}^{(j)}},
\end{align}
which proves \eqref{eq:Gamma_def_main}.

\vspace{-2mm}
\section{Proof of Proposition~\ref{prop:cond_cdf_sir_main}}\label{app:per_round_sir_proof}
Since all HARQ rounds are statistically identical, we omit the round index and consider any given round. 
We use the block partition \(\{\mathcal K_b\}_{b=1}^{B}\) of the \(K\) FAS ports defined in the spatial block-correlation model in Section~\ref{NetMod}, where \(\mathcal K_b\) denotes the set of port indices belonging to block \(b\), with \(|\mathcal K_b|=L_b\).

For simplicity, define
\begin{equation}\label{eq:kappa_def_app}
\kappa \triangleq \frac{\mu}{\sqrt{1-\mu^2}},
\end{equation}
so that $\mu^2/(1-\mu^2)=\kappa^2$.

{\em Desired-signal term}---For any port $k\in\mathcal K_b$, the normalized desired-signal power in \eqref{eq:X_ukn_def} can be written as
\begin{equation}\label{eq:S_def_app}
S_{u,k}=\bigl(x_{u,k}+\kappa x_{u,b}\bigr)^2+\bigl(y_{u,k}+\kappa y_{u,b}\bigr)^2.
\end{equation}

We define the common-component energy of group~$b$ as $r_{u,b}\triangleq x_{u,b}^2+y_{u,b}^2$, for $b=1,\ldots,B$. Since $x_{u,b},y_{u,b}\sim\mathcal N(0,\tfrac12)$ are independent, $r_{u,b}\sim \mathrm{Exp}(1)$, with its probability density function (PDF) given by
\begin{equation}\label{eq:rb_pdf_app}
f_{r_{u,b}}(r)=e^{-r},~r\ge 0.
\end{equation}
Note that $\{r_{u,b}\}_{b=1}^{B}$ are mutually independent. Conditioned on $r_{u,b}$, $2S_{u,k}$ follows non-central chi-square distribution, i.e., $\big(2S_{u,k}\,\big|\,r_{u,b}\big)\sim\chi'^2_2\!\bigl(2\kappa^2 r_{u,b}\bigr)$. Hence, the conditional PDF of $S_{u,k}$ is given by
\begin{equation}\label{eq:fS_cond_app}
f_{S_{u,k}\mid r_{u,b}}(s)=e^{-(s+\kappa^2 r_{u,b})}I_0\bigl(2\kappa\sqrt{r_{u,b}s}\bigr),~s\ge 0,
\end{equation}
and its conditional CDF is given by
\begin{equation}\label{eq:FS_cond_app}
F_{S_{u,k}\mid r_{u,b}}(t)=1-Q_1\!\left(\sqrt{2\kappa^2 r_{u,b}},\sqrt{2t}\right),~t\ge 0,
\end{equation}
where $I_0(\cdot)$ denotes the zero-order modified Bessel function of the first kind and $Q_1(\cdot,\cdot)$ is the first-order Marcum-$Q$ function.

Under the adopted latent-variable representation, the statistical dependence among $\{S_{u,k}\}_{k\in\mathcal K_b}$ is induced solely by the common variable $r_{u,b}$. Therefore, conditioned on $r_{u,b}$, these variables are independent within group~$b$. Averaging the resulting conditional product form over $\{r_{u,b}\}_{b=1}^{B}$ yields
\begin{equation}
\label{eq:jointS_cdf_app}
\begin{aligned}
F_{\{S_{u,k}\}_{k=1}^{K}}(t_1,\ldots,t_K)
&=
\prod_{b=1}^{B}
\int_{0}^{\infty}
e^{-r_{u,b}}
\\
&\quad\times
\prod_{k\in\mathcal K_b}
F_{S_{u,k}\mid r_{u,b}}(t_k)
\,dr_{u,b}.
\end{aligned}
\end{equation}

{\em Interference Term}---Considering \eqref{eq:I_u_k_i_sum_short} and conditioning on $\bar A=m$, let $\mathcal A\triangleq\{\tilde u\neq u:\, a_{\tilde u}=1\}$ denote the set of active interferers in the considered round, so that $|\mathcal A|=m$. Then, for any port $k\in\mathcal K_b$, the aggregate interference power can be written as $I_{u,k}=\sum_{\tilde u\in\mathcal A} S_{\tilde u,k}$. For each group $b$, define
$\tilde r_{u,b}\triangleq\sum_{\tilde u\in\mathcal A}\bigl(x_{\tilde u,b}^2+y_{\tilde u,b}^2\bigr)$. Conditioned on $\bar A=m$, $\big(\tilde r_{u,b}\mid \bar A=m \big)\sim \mathrm{Gamma}(m,1)$, with PDF given by
\begin{equation}\label{eq:tilder_pdf_app}
f_{\tilde r_{u,b}\mid \bar A=m}(r)=\frac{r^{m-1}e^{-r}}{\Gamma(m)},~r\ge 0.
\end{equation}
Moreover, $\big(2I_{u,k}\,\big|\,\tilde r_{u,b},\bar A=m \big)\sim\chi'^2_{2m}\!\bigl(2\kappa^2\tilde r_{u,b}\bigr)$, for $k\in\mathcal K_b$. Thus, the conditional PDF of $I_{u,k}$ is given by
\begin{equation}\label{eq:fI_cond_app}
\begin{aligned}
f_{I_{u,k}\mid \tilde r_{u,b},\bar A=m}(z)
&=
e^{-(z+\kappa^2\tilde r_{u,b})}
\left(
\frac{z}{\kappa^2\tilde r_{u,b}}
\right)^{\frac{m-1}{2}}
\\
&\quad\times
I_{m-1}\!\bigl(2\kappa\sqrt{\tilde r_{u,b}z}\bigr).
\end{aligned}
\end{equation}
Under the same latent-variable representation, conditioning on $(\tilde r_{u,b},\bar A=m)$ renders the variables $\{I_{u,k}\}_{k\in\mathcal K_b}$ independent within group~$b$. Hence,
\begin{equation}
\label{eq:jointI_pdf_app}
\begin{aligned}
&f_{\{I_{u,k}\}_{k=1}^{K}\mid \bar A=m}(z_1,\ldots,z_K)
\\
&\quad=
\prod_{b=1}^{B}
\Biggl[
\int_{0}^{\infty}
\frac{\tilde r_{u,b}^{m-1}e^{-\tilde r_{u,b}}}{\Gamma(m)}
\prod_{k\in\mathcal K_b}
f_{I_{u,k}\mid \tilde r_{u,b},\bar A=m}(z_k)\,
d\tilde r_{u,b}
\Biggr].
\end{aligned}
\end{equation}

{\em Per-Round SIR}---The selected SIR in the considered round is $\mathrm{SIR}_u=\max_{k\in\{1,\ldots,K\}} S_{u,k}/I_{u,k}$, with HARQ round index omitted. Moreover, the desired-signal vector $\{S_{u,k}\}_{k=1}^{K}$ is independent of the interference vector $\{I_{u,k}\}_{k=1}^{K}$. Hence, conditioned on $\bar A=m$, the CDF of ${\rm SIR}_u$ can be derived as 
\begin{equation}\label{eq:FSIR_step2_app}
\begin{aligned}
F_{\mathrm{SIR}_u}&(\gamma\mid \bar A=m) \\
= &
\mathbb P\!\left(
\frac{S_{u,1}}{I_{u,1}}<\gamma,\ldots,\frac{S_{u,K}}{I_{u,K}}<\gamma
\,\middle|\, \bar A=m
\right)
\\
= &
\int_{\mathbb R_+^K}
F_{\{S_{u,k}\}_{k=1}^{K}}(\gamma z_1,\ldots,\gamma z_K)
\\
&\times
f_{\{I_{u,k}\}_{k=1}^{K}\mid \bar A=m}(z_1,\ldots,z_K)\,
dz_1\cdots dz_K.
\end{aligned}
\end{equation}
Substituting \eqref{eq:jointS_cdf_app} and \eqref{eq:jointI_pdf_app} into \eqref{eq:FSIR_step2_app}, we obtain 
\begin{equation}\label{eq:FSIR_step3_app}
\begin{aligned}
F_{\mathrm{SIR}_u}\!(\gamma\!\mid\! \bar A\!=\!m)\!
=&
\prod_{b=1}^{B}
\int_{0}^{\infty}\!\!\int_{0}^{\infty}
e^{-r_{u,b}}
\frac{\tilde r_{u,b}^{m-1}e^{-\tilde r_{u,b}}}{\Gamma(m)}
\\
&\times
\bigl[\Xi_m(\gamma;r_{u,b},\tilde r_{u,b})\bigr]^{L_b}
\,dr_{u,b}\,d\tilde r_{u,b},
\end{aligned}
\end{equation}
where
\begin{equation}
\label{eq:Xi_integral_form_app}
\begin{aligned}
\Xi&_m(\gamma;r_{u,b},\tilde r_{u,b})\\
&\!\triangleq\!
\int_{0}^{\infty}
F_{S_{u,k}\mid r_{u,b}}(\gamma z)\,
f_{I_{u,k}\mid \tilde r_{u,b},\bar A=m}(z)\,dz
\\
&\!=\!
1\!-\!\!\!
\int_{0}^{\infty}\!\!\!\!
Q_1\left(\!\sqrt{2\kappa^2 r_{u,b}},\sqrt{2\gamma z}\right)
\!\times\!
f_{I_{u,k}\mid \tilde r_{u,b},\bar A=m}(z)\,dz.
\end{aligned}
\end{equation}
For fixed $(r_{u,b},\tilde r_{u,b},m)$, the quantity $\Xi_m(\gamma;r_{u,b},\tilde r_{u,b})$ is identical for all $k\in\mathcal K_b$. Thus, the product over the $L_b$ ports in group~$b$ reduces to $\bigl[\Xi_m(\gamma;r_{u,b},\tilde r_{u,b})\bigr]^{L_b}$. Finally, evaluating the integral in \eqref{eq:Xi_integral_form_app} in closed form yields \eqref{eq:Xi_main}. Substituting \eqref{eq:Xi_main} into \eqref{eq:FSIR_step3_app} gives \eqref{eq:FSIR_conditional_main}, which completes the proof.

\vspace{-2mm}
\section{Proof of Proposition~\ref{prop:avg_rounds_moments}}\label{app:proof_avg_transmissions}
The variable $\bar C$ is an integer within $\{1,\ldots,C\}$. Its expectation can be calculated as
\begin{equation}\label{eq:ECbar_app_start}
\mathbb E[\bar C]=\sum_{k=1}^{C} k\,\mathbb P(\bar C=k)=\sum_{k=1}^{C}\sum_{i=1}^{k}\mathbb P(\bar C=k).
\end{equation}
Reordering the summation, it can be rewritten as
\begin{equation}
\label{eq:ECbar_app_swap}
\mathbb E[\bar C]=\sum_{i=1}^{C}\sum_{k=i}^{C}\mathbb P(\bar C=k)=\sum_{i=1}^{C}\mathbb P(\bar C\ge i),
\end{equation}
which is the tail-sum identity. For $i=2,\ldots,C$, the event $\{\bar C\ge i\}$ is equivalent to $\{\Gamma_u^{(q)}(i-1)<\gamma_{\mathrm{th}}\}$. Hence,
\begin{equation}\label{eq:ECbar_app_tail}
\mathbb E[\bar C]=1 + \sum_{i=2}^{C}\mathbb P\!\left(\Gamma_u^{(q)}(i\!-\!1)\!<\!\gamma_{\mathrm{th}}\right)=1 + \sum_{j=1}^{C-1}F_{\Gamma(j)}(\gamma_{\mathrm{th}}),
\end{equation}
which proves \eqref{eq:ECbar_tail_main_new}. Similarly, the tail-sum identity for the second moment of a positive integer-valued random variable gives
\begin{equation}\label{eq:EC2_app_start}
\mathbb E[\bar C^2]=\sum_{i=1}^{C}(2i-1)\,\mathbb P(\bar C\ge i).
\end{equation}
Re-indexing via $j=i-1$, we obtain
\begin{equation}\label{eq:EC2_app_final}
\mathbb E[\bar C^2]=1+\sum_{j=1}^{C-1}(2j+1)\,F_{\Gamma(j)}(\gamma_{\mathrm{th}}),
\end{equation}
which proves \eqref{eq:EC2_tail_main_new}.

\bibliographystyle{IEEEtran}

\end{document}